\documentclass[sn-basic,Numbered]{sn-jnl}

\usepackage{graphicx}\usepackage{multirow}\usepackage{array}
\usepackage{amsmath,amssymb,amsfonts}\usepackage{amsthm}\usepackage{mathrsfs}\usepackage[title]{appendix}\usepackage{xcolor}\usepackage{textcomp}\usepackage{manyfoot}\usepackage{booktabs}\usepackage{algpseudocode}\usepackage{listings}\usepackage{complexity}
\usepackage{tikz}
\usepackage{algorithm2e}
\usetikzlibrary {arrows,graphs,patterns,decorations,calligraphy}
\usepackage{comment}

\DeclareMathOperator{\rl}{\mathrm{rl}}
\DeclareMathOperator{\td}{\mathrm{td}}
\DeclareMathOperator{\ctd}{\mathrm{ctd}}
\DeclareMathOperator{\lca}{\mathrm{lca}}
\DeclareMathOperator{\cl}{\uparrow\!}

\newcommand{\qedclaim}{\hfill $\diamond$ \medskip}

\newcommand{\bproblem}[3]{{\vspace*{0.06cm} \noindent
    \frame{\frame{\fbox{\begin{minipage}{0.9\textwidth}\textsc{{#1}}\vspace{0.1cm}\\
\textbf{Input:} {#2}\\\textbf{Question:}{ #3}\end{minipage}
	}}}}}

\newcommand{\kED}{\textsc{Fast-Strategy}}

\newcommand{\guards}{D}
\newcommand{\edn}{\gamma^{\infty}}
\newcommand{\medn}{\gamma_m^{\infty}}

\theoremstyle{thmstyleone}\newtheorem{theorem}{Theorem}\newtheorem{proposition}[theorem]{Proposition}\newtheorem{lemma}[theorem]{Lemma}
\newtheorem{corollary}[theorem]{Corollary}
\newtheorem{conjecture}[theorem]{Conjecture}

\newtheorem{definition}[theorem]{Definition}
\newtheorem{observation}[theorem]{Observation}
\newtheorem{example}[theorem]{Example}

\begin{document}

\title[Fast winning strategies for the attacker in eternal domination]{Fast winning strategies for the attacker in eternal domination\footnote{This research was supported by the ANR project P-GASE (ANR-21-CE48-0001-01).}}

\author[1]{\fnm{Guillaume} \sur{Bagan}}\email{guillaume.bagan@univ-lyon1.fr}

\author[1]{\fnm{Nicolas} \sur{Bousquet}}\email{nicolas.bousquet@univ-lyon1.fr}
\author[1]{\fnm{Nacim} \sur{Oijid}}\email{nacim.oijid@univ-lyon1.fr}
\author*[1]{\fnm{Théo} \sur{Pierron}}\email{theo.pierron@univ-lyon1.fr}

\affil[1]{\orgname{Universite Claude Bernard Lyon 1, CNRS, INSA Lyon}, \orgdiv{LIRIS, UMR5205} \orgaddress{\city{Villeurbanne}, \postcode{69622}, \country{France}}}

\abstract{Dominating sets in graphs are often used to model monitoring problems, by posting guards on the vertices of the dominating set. If an (unguarded) vertex is attacked, at least one guard can then react by moving there. This yields a new set of guards, which may not be dominating anymore. A dominating set is \emph{eternal} if one can endlessly resist to attacks.

From the attacker's perspective, if we are given a non-eternal dominating set, the question is to determine how fast can we provoke an attack that cannot be handled by a neighboring guard. We investigate this question from a computational complexity point of view, by showing that this question is \PSPACE-hard, even for graph classes where finding a minimum eternal dominating set is in \P. 

We then complement this result by giving polynomial time algorithms for cographs and trees, and showing a connection with tree-depth for the latter. We also investigate the problem from a parameterized complexity perspective, mainly considering two parameters: the number of guards and the number of steps.}


\maketitle

\section{Introduction}

Let $G = (V, E)$ be a graph and $\guards \subseteq V$ be a set of vertices called \emph{guards}.
The \emph{mobile domination game} is played by two players: Attacker and Defender.
At each turn, Attacker chooses a non-guarded vertex $v$ (we say that Attacker \emph{attacks}~$v$). Then, Defender has to move a guard along an edge $uv$ incident to $v$ (we say that Defender defends against the attack); see Figure~\ref{fig:edn} for an illustration. If she cannot move such a guard (that is $v$ had no neighbor $u$ in $\guards$), Attacker wins. Otherwise, the game continues with the set $\guards\cup\{v\}\setminus\{u\}$ of guards. If Defender can defend against any infinite sequence of attacks, the initial set of guards $\guards$ is called an \emph{eternal dominating set}. The \emph{eternal domination number} $\edn(G)$ \cite{infiniteorder2004} is the size of the smallest eternal dominating set of $G$. This notion has been widely studied, see e.g.~\cite{survey-eternal} for a recent survey. Many variants have been introduced in the literature, for example when all the guards can move at each attack~\cite{goddard}. This leads to similar notions of m-eternal dominating set and m-eternal domination number $\medn(G)$ (where m stands for ``multiple'' and is not a parameter).

\begin{figure}[!ht]
\centering
\scalebox{0.7}{
	\begin{tikzpicture}[scale=0.85,yscale=.7,vertex/.style={circle,draw,thick,fill=white,inner sep=2pt}]

\draw[>=triangle 45, ->] (3,0) -- (4,1);
\draw[>=triangle 45, ->] (3,-3) -- (4,-4);
\draw[>=triangle 45, ->] (10,-4) -- (11.5,-4);

\begin{scope}[shift={(0,-3)}]
    \node[vertex] (v0) at (-2,0) {\textcolor{white}{0}};
    \node[vertex] (v1) at (0,0) {\textcolor{white}{0}};
    \node[vertex,fill=gray!30] (v2) at (2,0) {3};
    \node[vertex,fill=gray!30] (v3) at (-2,2) {1};
    \node[vertex,fill=gray!30] (v4) at (0,2) {2};
    \node[vertex] (v5) at (2,2) {\textcolor{white}{0}};
    \node[vertex] (v6) at ((-1,3) {\textcolor{white}{0}};
    
    \draw[thick,color=red] (-1,3) circle(0.4);
    
    \path
    (v0) edge (v1)
    (v1) edge (v2) 
    (v0) edge (v3)
    (v1) edge (v4)
    (v2) edge (v5)
    (v3) edge (v4)
    (v3) edge (v6)
    (v4) edge (v6)
    ;
\end{scope}

\begin{scope}[shift={(7,0)}]
    \node[vertex] (v0) at (-2,0) {\textcolor{white}{0}};
    \node[vertex] (v1) at (0,0) {\textcolor{white}{0}};
    \node[vertex,fill=gray!30] (v2) at (2,0) {3};
    \node[vertex] (v3) at (-2,2) {\textcolor{white}{0}};
    \node[vertex,fill=gray!30] (v4) at (0,2) {2};
    \node[vertex] (v5) at (2,2) {\textcolor{white}{0}};
    \node[vertex,fill=gray!30] (v6) at ((-1,3) {1};
    
    \draw[thick,color=red] (-2,0) circle(0.4);
    
    \path
    (v0) edge (v1)
    (v1) edge (v2) 
    (v0) edge (v3)
    (v1) edge (v4)
    (v2) edge (v5)
    (v3) edge (v4)
    (v3) edge (v6)
    (v4) edge (v6)
    ;
\end{scope}

\begin{scope}[shift={(7,-5)}]
    \node[vertex] (v0) at (-2,0) {\textcolor{white}{0}};
    \node[vertex] (v1) at (0,0) {\textcolor{white}{0}};
    \node[vertex,fill=gray!30] (v2) at (2,0) {3};
    \node[vertex,fill=gray!30] (v3) at (-2,2) {1};
    \node[vertex] (v4) at (0,2) {\textcolor{white}{0}};
    \node[vertex] (v5) at (2,2) {\textcolor{white}{0}};
    \node[vertex,fill=gray!30] (v6) at ((-1,3) {2};
    
    \draw[thick,color=red] (0,0) circle(0.4);
    
    \path
    (v0) edge (v1)
    (v1) edge (v2) 
    (v0) edge (v3)
    (v1) edge (v4)
    (v2) edge (v5)
    (v3) edge (v4)
    (v3) edge (v6)
    (v4) edge (v6)
    ;
\end{scope}

\begin{scope}[shift={(14,-5)}]
    \node[vertex] (v0) at (-2,0) {\textcolor{white}{0}};
    \node[vertex,fill=gray!30] (v1) at (0,0) {3};
    \node[vertex] (v2) at (2,0) {\textcolor{white}{0}};
    \node[vertex,fill=gray!30] (v3) at (-2,2) {1};
    \node[vertex] (v4) at (0,2) {\textcolor{white}{0}};
    \node[vertex] (v5) at (2,2) {\textcolor{white}{0}};
    \node[vertex,fill=gray!30] (v6) at ((-1,3) {2};
    
    \draw[thick,color=red] (2,2) circle(0.4);
    
    \path
    (v0) edge (v1)
    (v1) edge (v2) 
    (v0) edge (v3)
    (v1) edge (v4)
    (v2) edge (v5)
    (v3) edge (v4)
    (v3) edge (v6)
    (v4) edge (v6)
    ;
\end{scope}

\begin{scope}[shift={(0,-8.5)}]
    \draw[pattern=north west lines, pattern color=blue!30] (-2.5,1.5) rectangle (0.5,3.5);
    \draw[pattern=north west lines, pattern color=blue!30] (-2.5,-0.5) rectangle (0.5,0.5);
    \draw[pattern=north west lines, pattern color=blue!30] (1.5,-0.5) rectangle (2.5,2.5);
    
    \node[vertex,fill=gray!30] (v0) at (-2,0) {1};
    \node[vertex] (v1) at (0,0) {\textcolor{white}{0}};
    \node[vertex,fill=gray!30] (v2) at (2,0) {2};
    \node[vertex] (v3) at (-2,2) {\textcolor{white}{0}};
    \node[vertex,fill=gray!30] (v4) at (0,2) {3};
    \node[vertex] (v5) at (2,2) {\textcolor{white}{0}};
    \node[vertex] (v6) at ((-1,3) {\textcolor{white}{0}};
        
    \path
    (v0) edge (v1)
    (v1) edge (v2) 
    (v0) edge (v3)
    (v1) edge (v4)
    (v2) edge (v5)
    (v3) edge (v4)
    (v3) edge (v6)
    (v4) edge (v6)
    ;
\end{scope}

	\end{tikzpicture}}
\caption{Above, an example of a winning strategy for Attacker in 3 turns, depending on the answer of Defender.
Each number represents a guard, and the attack is circled. Below, an eternal dominating set with 3 guards. Each guard protects its own clique.}
\label{fig:edn}
\end{figure}
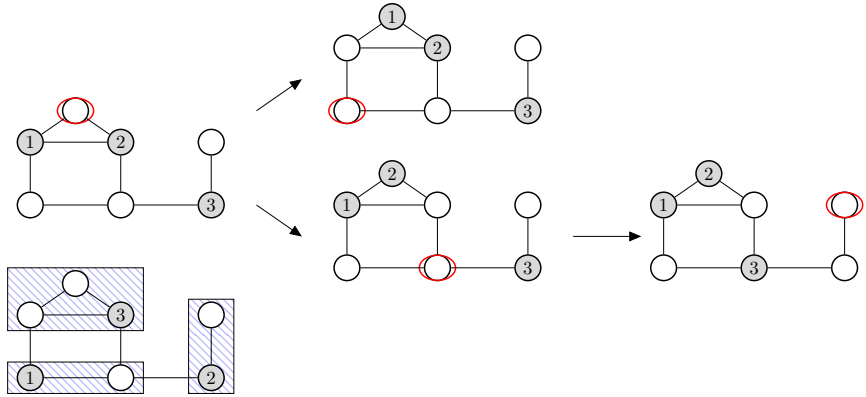

For every graph $G$, we have $\alpha(G) \leq \edn(G) \leq \theta(G)$~\cite{infiniteorder2004} and $\edn(G) \leq \binom{\alpha(G)+1}{2}$~\cite{fixed-independence}, where $\alpha(G)$ is the maximum size of an independent set in $G$, and $\theta(G)$ is the minimum size of a clique cover of $G$. The latter inequality has been proved to be tight for an infinite class of graphs~\cite{independence-tight}. Several existing works also focus on classes of graphs for which $\edn(G) = \theta(G)$, see e.g.~\cite{infiniteorder2004,eds-demand,eds-tree}.

One can naturally wonder if these parameters can be easily computed. This motivates the introduction of the following problem. 

\bproblem{Eternal-Dominating-Number}{A graph $G$, an integer $k \geq 1$}{
$\edn(G) \leq k$?}

This problem lies in \EXP~and is \coNP-hard~\cite{eds-digraph} but the exact complexity class this problem belongs to remains open. However, on restricted graph classes, there exist some polynomial time algorithms such as on perfect graphs. Indeed, for perfect graphs, we have $\alpha(G) = \edn(G) = \theta(G)$, which yields a polynomial algorithm for
\textsc{Eternal-Dominating-Number}~\cite{infiniteorder2004}.

Deciding the m-variant of \textsc{Eternal-Dominating-Number} (where every guard is allowed to move at each turn) is harder since deciding if $\medn(G) \leq k$ is already \NP-hard even on Hamiltonian split graphs \cite{meds-split}. 
However, polynomial time algorithms have been proposed on trees \cite{eds-tree}, unit interval graphs \cite{eds_proper_interval}, interval graphs \cite{meds-interval} and cactus graphs \cite{meds-cactus}.
Moreover, great attention has been paid to the study of lower and upper bounds for the m-eternal dominating number in grids. See for example \cite{lamprou-grid,fionn-grid}.
The spy game, a generalization of m-eternal domination has been proved \PSPACE-hard on directed graphs and \NP-hard on (non-directed) graphs \cite{spygame}. Another generalization of m-eternal domination, the guarding game, has been proved \EXP-complete~\cite{guarding-game}.

Surprisingly, the situation is quite different for the seemingly close problem:

\bproblem{Eternal-Dominating-Set}{A graph $G$, a set of guards $\guards \subseteq V$}{Is $\guards$ an eternal dominating set of $G$?}

This problem has been recently proved \EXP-complete \cite{virgelot2024}.
It is not clear that this problem is easier than the first. Indeed, there is nothing preventing graphs to have an easy-to-find eternal dominating set of size $k$ while it may be hard to tell whether some given vertices form an eternal dominating set. In particular, on perfect graphs, the complexity of \textsc{Eternal-Dominating-Set} is still open contrary to the complexity of \textsc{Eternal-Dominating-Number}.

In this paper, we are not only interested in determining whether a configuration of guards is an eternal dominating set but also to find, in the negative case, a winning strategy for Attacker minimizing the number of steps. This question naturally arises for board games, like Chess or Go, but has also attracted a lot of attention in the combinatorial game theory literature: for the domination game~\cite{complexity-domination-game},
maker-breaker games~\cite{makerbreakerdom,makerbreaker},
the non-planarity game~\cite{non-planarity} and the spy game~\cite{spygame}.
Recently, this problem has been studied for m-eternal domination on trees in~\cite{meds-efficient}, where Bla\v{z}ej et al. prove that if Defender does not have enough guards to m-eternally dominate a graph $G$, then Attacker can win in at most $d$ turns, where $d$ is the diameter of $G$.

Informally, we say that Attacker wins in $t$ turns on $(G,\guards)$ if he has a strategy that makes Defender lose before the $t$-th turn, regardless of her strategy. We denote by $t_G(\guards)$ the minimum $t$ such that Attacker wins in $t$ turns on $(G,\guards)$ (We postpone the formal definitions to the next section.) In particular, $t_G(\guards)$ is~$+\infty$ if and only if $\guards$ is an eternal dominating set of $G$. 

The goal of this paper is to study the complexity of computing $t_G(\guards)$, that is, of the following problem.

\bproblem{$\kED$}{A graph $G = (V, E)$, a set of guards $\guards \subseteq V$, an integer $t \geq 1$}{is $t_G(\guards) \leq t$?}

\paragraph{Our results and organization of the paper.}

In Section~\ref{sec:bip}, we show that this problem is \coNP-hard on bipartite and split graphs. On the positive side, we show that \textsc{Eternal-Dominating-Set} can be decided in polynomial time on bipartite graphs. 
In Section~\ref{sec:perfect}, we show that $\kED$ is \PSPACE-hard even restricted to perfect graphs (recall that \textsc{Eternal-Dominating-Number} can be decided in polynomial time on perfect graphs). More specifically, we prove  that the problem is \PSPACE-hard on $2$-unimodal graphs reducing from the \textsc{Unordered-CNF} problem. 

 Then, we give two positive results: we prove that $\kED$ can be decided in polynomial time on trees (Section~\ref{sec:tree}) and on cographs (Section~\ref{sec:cographs}). For trees, the main step of the proof consists in introducing arenas, which are special subtrees and proving that Attacker can win in at most $k$ steps if and only if there exists an arena for which some treedepth-related parameter is at most $k$. We finally prove that the polynomial time algorithm for computing the treedepth on trees can be adapted for this new parameter. For cographs, we use their decomposition as disjoint unions or joins of smaller cographs and analyze the strategies of Attacker and Defender in each case to provide a recursive algorithm for $\kED$. Finally, in Section~\ref{sec:param}, we study the parameterized complexity of $\kED$ for several classes of graphs.

\paragraph{Related work}
The notion of $k$-secure dominating set was introduced by Burger et al.~\cite{finiteorder2004}. In their setting, a dominating set is \emph{$k$-secure} if it remains dominating after a sequence of $k$ attacks, which looks like the problem we are considering in this paper. However, in their setting, the sequence of attacks is known in advance (oblivious adversary) whereas, in our case, Attacker can adapt his moves according to how Defender defends the attacks (adaptive adversary).

\section{Preliminaries} 

For an integer $n > 0$, $[n]$ represents the set $\{1, 2, \ldots, n\}$.
All graphs considered in this paper are finite, loopless, and simple. Let $G = (V, E)$ be a graph.
Given a vertex $v \in V$, $N(v)$ denotes the \emph{neighborhood} of $v$, i.e., the set $\{ y \in V: vy \in E \}$. The \emph{degree} $d(v)$ of $v$ is $|N(v)|$.
A set $S \subseteq V$ is an \emph{independent set} of $G$ if there is no edge $uv$ for every $u, v$ in $S$. The \emph{independence number} $\alpha(G)$ is the size of a largest independent set of $G$. A \emph{clique} is a subset of pairwise adjacent vertices. 
The \emph{clique covering number} $\theta(G)$ is the minimum number of cliques in which $V$ can be partitioned. A \emph{vertex cover} of $G$ is a set $S \subseteq V$ such that every edge $e$ in $E$ has an endpoint in $S$. A \emph{matching} of $G$ is a set of edges of $G$ that pairwise do not share an endpoint.
A graph $G$ is \emph{perfect} if $\alpha(H) = \theta(H)$ for every induced subgraph $H$ of $G$.

A \emph{rooted tree} $T$ is a tree with a distinguished vertex $r$ called the root.
Consider the orientation of $T$ such that every vertex is reachable from $r$.
$u$ is a \emph{child} of $v$ if there is an arc $vu$ in this orientation and $u$ is a \emph{descendant} of $v$ if there is an oriented path from $v$ to $u$ along this orientation. A rooted subtree of $T$ is a subtree of $T$ induced by a node and all its descendants. The \emph{height} of $T$ is the number of vertices in a longest directed path starting at $r$.

A \emph{td-decomposition} of a connected graph $G = (V, E)$ is a rooted tree $T$ with set of vertices $V$ such that $u$ is a descendant of $v$ or $v$ is a descendant of $u$ for every $uv \in E$. The \emph{treedepth} of $G$ is the minimum height of a td-decomposition of $G$ (see Figure~\ref{fig:td}).

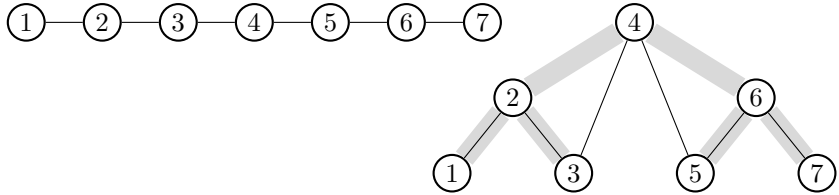
\begin{figure}[!ht]
\centering
\begin{tikzpicture}[vertex/.style={circle,draw,thick,fill=white,inner sep=2pt}]
\begin{scope}[shift={(-8,0)}]
\node[vertex] (c1) at (0,0) {1};
\node[vertex] (c2) at (1,0) {2};
\node[vertex] (c3) at (2,0) {3};
\node[vertex] (c4) at (3,0) {4};
\node[vertex] (c5) at (4,0) {5};
\node[vertex] (c6) at (5,0) {6};
\node[vertex] (c7) at (6,0) {7};
\path
    (c1) edge (c2)
    (c2) edge (c3)
    (c3) edge (c4)
    (c4) edge (c5)
    (c5) edge (c6)
    (c6) edge (c7)
    ;
\end{scope}

\begin{scope}[shift={(0,0)},xscale=.8]
\node[vertex] (v4) at (0,0) {4};
\node[vertex] (v2) at (-2,-1) {2};
\node[vertex] (v6) at (2,-1) {6};
\node[vertex] (v1) at (-3,-2) {1};
\node[vertex] (v3) at (-1,-2) {3};
\node[vertex] (v5) at (1,-2) {5};
\node[vertex] (v7) at (3,-2) {7};
    \path
    (v4) edge[line width=8pt,color=gray!30] (v2)
    (v4) edge[line width=8pt,color=gray!30] (v6) 
    (v2) edge[line width=8pt,color=gray!30] (v1)
    (v2) edge[line width=8pt,color=gray!30] (v3)
    (v6) edge[line width=8pt,color=gray!30] (v5)
    (v6) edge[line width=8pt,color=gray!30] (v7)
    (v1) edge (v2)
    (v2) edge (v3)
    (v3) edge (v4)
    (v4) edge (v5)
    (v5) edge (v6)
    (v6) edge (v7)
    ;
    \end{scope}
\end{tikzpicture}
\caption{On the left, the graph $P_7$ of treedepth 3. On the right, an optimal td-decomposition of $P_7$}
\label{fig:td}
\end{figure}

Let $G = (V, E)$ be a graph and $\guards \subseteq V$ be a set of $g$ vertices called \emph{guards}. In the eternal domination game, a \emph{strategy for Attacker} is a function mapping each set $S$ of $g$ vertices to a vertex of $V\setminus S$. A \emph{strategy for Defender} is a partial function mapping pairs $(S, v)$ where $S$ is a set of $g$ vertices and $v\in V\setminus S$ to some vertex of $N(v)\cap S$ if it exists. 

Given two strategies $f,g$ for Attacker and Defender, we say that Defender can \emph{resist $t$ attacks} on $(G,\guards)$ using $(f,g)$ if either $t=0$, or $g(\guards, f(\guards))$ is defined, and she can resist $t-1$ turns on $(G,\guards\cup\{g(\guards,f(\guards))\}\setminus\{f(\guards)\})$ using $(f,g)$. We extend this notion: Defender can resist $t$ attacks on $(G,\guards)$ when Attacker plays according to $f$ if she has a strategy $g$ so that the previous condition is satisfied.

We are interested in the \emph{fastest way} for Attacker to win (if he can), that is the minimum $t$ such that Attacker has a strategy $f$ such that Defender cannot resist $t$ attacks on $(G,\guards)$ when Attacker plays according to $f$. We denote this number $t$ by $t_G(\guards)$. Note that $t_G(\guards)$ might be $+\infty$ if Defender can eternally answer to Attacker's move, \emph{i.e.} if Attacker has no winning strategy.

\section{Bipartite and split graphs}
\label{sec:bip}
The goal of this section is to study the complexity of the problems \textsc{Eternal-Dominating-Set} and $\kED$ on bipartite graphs. Recall that bipartite graphs are perfect, hence determining the smallest eternal dominating set can be done in polynomial time~\cite{goddard}. We extend this result to \textsc{Eternal-Dominating-Set}.

\begin{theorem}\label{eds-bipartite}
\textsc{Eternal-Dominating-Set} is in \P{} on bipartite graphs.
\end{theorem}

On the opposite, we show that determining the smallest number of moves in a winning strategy is significantly harder, as summarized by the following statement.

\begin{theorem}\label{conp-hard-bipartite}
$\kED$ is in \PSPACE{} and \coNP-hard on bipartite graphs.
\end{theorem}

Both these statements are based on the following lemma, which gives some insights about the shape of winning strategies for the eternal domination game on bipartite graphs. 

\begin{lemma}\label{bipartite-oneside}
Let $G$ be a bipartite graph and $\guards$ be a set of guards.
If Attacker wins in $t$ turns on $(G, \guards)$, then he can win in at most $t$ turns by always attacking the same side of the bipartition.
\end{lemma}

\begin{proof}
We prove this result by induction on $t$. The result holds when $t=1$. Hence, we can assume that $t>1$. 

Denote by $(A,B)$ a bipartition of $G$. Let $v$ be the first move of Attacker in a winning strategy. By symmetry, we can assume that $v\in A$. Let $u \in N(v) \cap \guards$ be the answer of Defender. Let us denote by $\guards' = (\guards\setminus\{u\})\cup\{v\}$ the new set of guards. 

By assumption on $v$, Attacker has a winning strategy in $t-1$ steps on $(G, \guards')$. By induction, there exists such a strategy consisting in playing only on $A$ or only on $B$. In the first case, the conclusion follows. So we may assume that Attacker has a winning strategy in $t-1$ turns on $(G,\guards')$ by playing only in $B$. 

Assume for contradiction that, in $(G,\guards)$, Defender can resist to $t-1$ attacks on $B$. Then she can also resist to $t-1$ attacks on $B$ in $(G,\guards')$ by answering $v$ if Attacker plays on $u$, and answering with a guard from $\guards\cap A$ otherwise, following her strategy on $(G,\guards)$.

Since Attacker was supposed to win on $t-1$ steps on $(G,D')$, this is a contradiction. Therefore Attacker can win in $t-1$ turns on $(G,\guards)$ by playing only on $B$, which concludes.
\end{proof}

The main consequence of this result is the following preprocessing statement.

\begin{lemma}\label{guardsEqOnePart}
Let $G$ be a bipartite graph and $\guards$ be a set of guards. Then, one can compute in polynomial time a bipartite graph $G'$ on the same set of vertices such that $t_G(\guards) = t_{G'}(\guards)$ where $\guards$ and $V(G')\setminus\guards$ form a bipartition of $V(G')$.
\end{lemma}

\begin{proof}
Let $G'$ be the graph obtained from $G$ by removing the edges between two guarded vertices, and between two unguarded vertices. $G'$ is a bipartite subgraph of $G$ with bipartition $(\guards,V(G)\setminus \guards)$. 

It remains to show that $t_G(\guards)=t_{G'}(\guards)$. Since $G'$ is a subgraph of $G$, any winning strategy of Attacker on $(G,\guards)$ is also winning in $(G',\guards)$, hence $t_G(\guards)\geqslant t_{G'}(\guards)$. For the converse direction, let $A\cup B$ be a bipartition of $G$. In particular, it is also a bipartition of $G'$, and by Lemma~\ref{bipartite-oneside} there exists a winning strategy for Attacker in $(G',\guards)$ in $t_{G'}(\guards)$ turns which consists of only playing in the same side of the bipartition, say $B$. We claim that following this strategy, Attacker wins in $t_{G'}(\guards)$ turns on $(G,\guards)$. More precisely, we claim that whenever Attacker attacks a vertex, Defender has the same possible answers in $G$ and $G'$. 

Since Attacker never plays in $A$, guards are only moved from $A$ to $B$. In particular, guards on $B$ cannot move and unguarded vertices in $A$ stay unguarded during the whole game. Therefore, if Attacker attacks $b\in B$, then $b\notin \guards$. Moreover, if $a\in A$ is a guarded neighbor of $b$ in $G$ then $a\in \guards$, hence $ab\in E(G')$.

Therefore, Attacker wins in $t_{G'}(\guards)$ turns on $(G,\guards)$, hence $t_G(\guards)=t_{G'}(\guards)$.
\end{proof}

We are now ready to prove Theorem~\ref{eds-bipartite}.

\begin{proof}[Proof of Theorem~\ref{eds-bipartite}]
Let $G = ((A,B), E)$ be a bipartite graph and $\guards$ be a set of guards. By Lemma~\ref{guardsEqOnePart}, we can modify the instance in order to assume that $B=\guards$. We claim that $\guards$ is an eternal dominating set of $G$ if and only if
$G$ admits a matching saturating $A$ (i.e. such that all the vertices of $A$ are matched). Proving this claim is enough to conclude since this criterion can be tested in polynomial time.

Let us now prove the claim. If $G$ has a matching $M$ saturating $A$ then the edges of $M$ and the unmatched vertices in $B$ form a clique cover of $G$ where each clique contains a guarded vertex. Hence, $G$ can be eternally dominated (Defender always defends an attacked vertex $v$ with a guard in the same clique as $v$ in the clique cover).

For the converse direction, assume there is no matching saturating $A$. By Hall's marriage theorem (see e.g.\cite{diestel-book}), there is a set of vertices $S \subseteq A$ such that $|N(S)| < |S|$. In particular, Attacker wins in at most $|N(S)|+1$ turns by successively attacking vertices in $S$: they can only be defended by vertices of $N(S)$. 
\end{proof}

The rest of this section is devoted to the proof of Theorem~\ref{conp-hard-bipartite}. The fact that is belongs to \PSPACE{} is actually an easy consequence of Lemma~\ref{guardsEqOnePart}, since it gives a polynomial bound (namely the number of guards plus one) on the minimum number of turns of a winning strategy for Attacker. 

It remains to show that $\kED$ is \coNP-hard on bipartite graphs.
To this end, we reduce the problem $\textsc{Independent-Set}$ to co-$\kED$. 

\bproblem{Independent set}{A graph $G$, an integer $k \geq 1$}{
$\alpha(G) \leq k$?}

Let $G$ be a graph and $k$ be an integer such that $(G,k)$ is an instance of $\textsc{Independent-Set}$. We assume that vertices of $G$ are labeled by integers from $[n]$. We build an instance $(G', \guards , t)$ of co-$\kED$ as follows (see Figure \ref{fig:bipartite}):
\begin{enumerate}
\item We create in $G'$ a graph $K_{k,n}$, and denote by $U=\{u_1,\ldots,u_k\}$ and $V=\{v_1,\ldots,v_n\}$ the vertices of each part.
\item For each edge $e \in E(G)$, we create a new complete bipartite graph $K_{k+1,k}$, with bipartition $S_e,T_e$.
\item For each edge $e = ij \in E(G)$, we connect each vertex of $S_e$ to $v_i$ and $v_j$.
\item The set $\guards$ of guards contains $V$ and all the $T_e$'s.
\end{enumerate}
Let $t = 2k+1$.
The graph $G'$ is bipartite and can be constructed in polynomial time. Now, Theorem~\ref{conp-hard-bipartite} boils down to prove the following.

\begin{figure}[!ht]
\begin{center}
\scalebox{0.8}{
	\begin{tikzpicture}[
every edge/.style = {draw=black,very thick},
 vertex/.style args = {#1/#2}{circle, draw, thick, fill=black,
      label=#1:#2}
                    ]

    \draw [dashed,color=red] (-1,1) -- (13,1);

    \node[circle, draw, thick,vertex=left/$u_1$] (u1) at (0.5, 0) {};
    \node[circle, draw, thick,vertex=left/$u_2$] (u2) at (1.5, 0) {};
    \node[circle, draw, thick,vertex=left/$u_3$] (u3) at (2.5, 0) {};

    \node[vertex=left/$v_1$] (v1) at (0, 2) {};
    \node[vertex=left/$v_2$] (v2) at (1, 2) {};
    \node[vertex=left/$v_3$] (v3) at (2, 2) {};
    \node[vertex=left/$v_4$] (v4) at (3, 2) {};

    \node at (1.5, -1){$k$};

    \draw [ultra thick,decorate,decoration={calligraphic brace,mirror}] (0.1, -0.5) --  (2.9,-0.5);

\path
    (u1) edge (v1)
    (u1) edge (v2)
    (u1) edge (v3)
    (u1) edge (v4)
        (u2) edge (v1)
    (u2) edge (v2)
    (u2) edge (v3)
    (u2) edge (v4)
        (u3) edge (v1)
    (u3) edge (v2)
    (u3) edge (v3)
    (u3) edge (v4)
    ;

    \node[circle, draw, thick] (s1) at (5, 0) {};
    \node[circle, draw, thick] (s2) at (6, 0) {};
    \node[circle, draw, thick] (s3) at (7, 0) {};
    \node[circle, draw, thick] (s4) at (8, 0) {};

    \node[circle, draw, thick,fill=black] (t1) at (5.5, 2) {};
    \node[circle, draw, thick,fill=black] (t2) at (6.5, 2) {};
    \node[circle, draw, thick,fill=black] (t3) at (7.5, 2) {};

    \node at (6.5, -1){$S_{1,2}$: $k+1$ vertices};

    \draw [ultra thick,decorate,decoration={calligraphic brace,mirror}] (4.6, -0.5) --  (8.4,-0.5);

    \node at (6.5, 3){$T_{1,2}$: $k$ vertices};

    \draw [ultra thick,decorate,decoration={calligraphic brace}] (5, 2.5) --  (8,2.5);

\path
    (s1) edge (t1)
    (s1) edge (t2)
    (s1) edge (t3)
        (s2) edge (t1)
    (s2) edge (t2)
    (s2) edge (t3)
        (s3) edge (t1)
    (s3) edge (t2)
    (s3) edge (t3)
        (s4) edge (t1)
    (s4) edge (t2)
    (s4) edge (t3)
    ;

\path
    (v1) edge (s1)
    (v1) edge (s2)
    (v1) edge (s3)
    (v1) edge (s4)
    (v2) edge (s1)
    (v2) edge (s2)
    (v2) edge (s3)
    (v2) edge (s4)
    ;

    \node[circle, draw, thick] (s1b) at (9, 0) {};
    \node[circle, draw, thick] (s2b) at (10, 0) {};
    \node[circle, draw, thick] (s3b) at (11, 0) {};
    \node[circle, draw, thick] (s4b) at (12, 0) {};

    \node[circle, draw, thick,fill=black] (t1b) at (9.5, 2) {};
    \node[circle, draw, thick,fill=black] (t2b) at (10.5, 2) {};
    \node[circle, draw, thick,fill=black] (t3b) at (11.5, 2) {};

\path
    (s1b) edge (t1b)
    (s1b) edge (t2b)
    (s1b) edge (t3b)
        (s2b) edge (t1b)
    (s2b) edge (t2b)
    (s2b) edge (t3b)
        (s3b) edge (t1b)
    (s3b) edge (t2b)
    (s3b) edge (t3b)
        (s4b) edge (t1b)
    (s4b) edge (t2b)
    (s4b) edge (t3b)
    ;

    \node at (10.5, -1){$S_{2,3}$: $k+1$ vertices};

    \draw [ultra thick,decorate,decoration={calligraphic brace,mirror}] (8.6, -0.5) --  (12.4,-0.5);

    \node at (10.5, 3){$T_{2,3}$: $k$ vertices};

    \draw [ultra thick,decorate,decoration={calligraphic brace}] (9, 2.5) --  (12,2.5);
    
 	\end{tikzpicture}
  }

\caption{
A partial representation of the reduction from  $(G, k)$
where $G$ is a graph with 4 vertices $\{v_1, v_2, v_3, v_4\}$ and two edges $v_1v_2$ and $v_2v_3$ and $k = 3$.
Every vertex in $S_{2,3}$ is connected to $v_2$ and $v_3$.
The vertices above the dashed line are guarded.
The length of the game is $t = 2k+1=7$.
}
\label{fig:bipartite}
\end{center}
\end{figure}
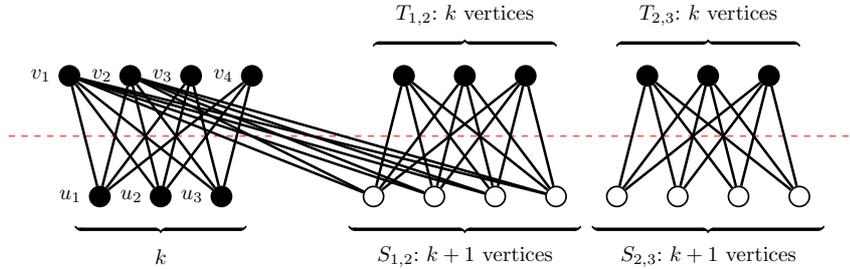

\begin{lemma}
\label{lem:red-conp}
    The graph $G$ has no independent set of size $k$ if and only if Attacker wins in at most $t$ turns on $(G', \guards)$.
\end{lemma}

\begin{proof}
    If $G$ has no independent set of size $k$, then, Attacker can win as follows. Attacker first attacks all the vertices of $U$ one by one. Defender must choose $k$ vertices $v_{j_1}, \ldots, v_{j_k}$ in $V$ to defend these attacks. Since $G$ has no independent set of size $k$, Defender must have chosen two vertices $v_i,v_j$ of $V$ that are connected by an edge $e=v_iv_j$ in $G$. Attacker then attacks the $k+1$ vertices of $S_e$ one after another. Since both $v_i$ and $v_j$ have been slid on $U$, Defender can only defend these attacks by sliding guards from $T_e$ to $S_e$. Since $S_e$ has size $k+1$ and $T_e$ only has size $k$, Attacker wins in at most $t=2k+1$ turns.

For the converse direction, assume that $G$ has an independent set $S = \{j_1, \ldots, j_k\}$ of size $k$. By Lemma~\ref{bipartite-oneside}, we can assume that Attacker only plays on $U$ or vertices in $S_e$ (since the other side of the bipartition only contains guarded vertices). We claim that Defender can defend against Attacker during $t$ turns with the following strategy:
\begin{enumerate}
\item If Attacker plays on $u_i\in U$, Defender moves the guard $v_{j_i}$.
\item If Attacked attacks a vertex in $S_e$ and that at least one other vertex of $S_e$ has not been attacked, then Defender moves a guard from a vertex of $T_e$.
\item If Attacker and attacks a vertex in $S_e$ with $e=v_iv_j$ and all the other vertices of $S_e$ have been attacked then Defender moves a guard from $\{ v_i,v_j\} \setminus S$ to defend the attack if such a vertex exists.
\end{enumerate}
This strategy indeed permits to defend against all the attacks on $U$. So if Defender cannot defend, it is at Step 3. Since $S$ is an independent set, $\{ v_i,v_j\} \setminus S$ is not empty for every edge and then Defender can defend against the attack of the last vertex of $S_e$ at least once. So Attacker cannot win in $t$ turns on $(G',\guards)$.
\end{proof}

This completes the proof of Theorem~\ref{conp-hard-bipartite}.
We can actually easily adapt this reduction to obtain the same result on split graphs.

\begin{theorem}\label{conp-hard-split}
$\kED$ is \coNP-hard on split graphs.
\end{theorem}

Let $(G',\guards)$ be the instance obtained by the previous reduction. We define the graph $G''$ as the graph obtained from $G'$ by adding a new vertex $s$ and all edges between $s$, $V$ and the vertices of each $T_e$. Fix also $\guards'=\guards\cup\{s\}$. Now we get the following analogue of Lemma~\ref{lem:red-conp}.

\begin{lemma}
$G$ has no independent set of size $k$ if and only if Attacker wins in at most $t$ turns on $(G'', \guards')$.
\end{lemma}

\begin{proof}
If $G$ has no independent set of size $k$, then Attacker wins in $t$ turns on $(G'',\guards')$ using the same strategy as in Lemma~\ref{lem:red-conp}.

Conversely, if $G$ has an independent set of size $k$, then Defender can also resist $t$ turns with a very similar strategy to Lemma~\ref{lem:red-conp}, except that if Attacker plays on $V\cup\{s\}$ or some $T_e$, she then answers by moving the guard initially posted on $s$.
\end{proof}

\section{Perfect graphs}
\label{sec:perfect}
Let us now prove that $\kED$ is \PSPACE-hard on \emph{$2$-unipolar graphs}. A graph is \emph{unipolar} if its vertices can be bipartitioned into a clique $V_1$ and a disjoint union of cliques $V_2$. It is moreover \emph{$k$-unipolar} if the cliques in $V_2$ have size at most $k$. Since unipolar graphs are perfect~\cite{unipolar-perfect}, $\kED$ also is \PSPACE-hard on perfect graphs. Moreover, $2$-unipolar are weakly chordal, hence the problem also is actually \PSPACE-complete on weakly chordal graphs.

\begin{theorem}\label{unipolar-pspace}
$\kED$ is \PSPACE-hard on $2$-unipolar graphs. \end{theorem}

The rest of this section is devoted to prove Theorem~\ref{unipolar-pspace}.
We make a reduction from \textsc{Unordered-CNF}.
An instance of \textsc{Unordered-CNF} consists of two sets of variables $X,Y$ with $|X|=|Y|$ and a CNF formula $\varphi$ with variables $X \cup Y$. The output is true if and only if Satisfier has a winning strategy in the following game: Two players, called Falsifier and Satisfier, successively choose a variable (Falsifier in $X$ and Satisfier in $Y$) and assigns it a truth value, until no variable remains. Satisfier wins if the assignment satisfies the formula $\varphi$.
Schaefer proves that \textsc{Unordered-CNF} is \PSPACE-complete \cite{schaefer78}.

\paragraph{Our reduction}
Let $(X,Y,\varphi)$ be an instance of \textsc{Unordered-CNF} and let $k=|X|=|Y|$ and $M=8k^2$. Without loss of generality, we can assume that $k\geqslant 2$ and every clause of $\varphi$ contains at least one variable from $Y$. We create an instance $(G_\varphi,\guards_\varphi)$ of $\kED$ as follows (see Figure~\ref{fig:pspace} for an illustration).
Let us denote by $\overline{X}$ (resp. $\overline{Y}$) the set $\{ \overline{x},x \in X\}$ (resp. $ \{\overline{y},y \in Y\}$).

\begin{itemize}
\item For each $x\in X$, we create two adjacent vertices $u_x$ and $u_{\bar{x}}$ in $G_\varphi$. 

Let us denote by $U$ the union of all these vertices.
\item For each $x\in X$ and $y\in Y$, we create four vertices $v_{x,y},v_{x,\bar{y}},v_{\bar{x},y}$ and $v_{\bar{x},\bar{y}}$, that are all in $\guards_\varphi$. \\
For every $y \in Y$, we denote by $V_y^* = \{ v_{a,b}$ with $b \in \{ y,\overline{y} \}\}$. Note that $V_y^*$ has size $4k$.
\item For every $a \in X \cup \overline{X}$ and $b \in Y \cup \overline{Y}$, we create an edge between $u_a$ and $v_{a,b}$.
\item For each variable $y\in Y$, we create two new sets of vertices $V_y,V'_y$ of size respectively $M$ and $M-4k+1$. We add a complete bipartite graph between $V_y$ and $V'_y$ called the \emph{variable checker of $y$}.

We add all the possible edges between $V_y$ and $V_y^*$ and we put guards on all the vertices of $V_y'$.
\item For each clause $C_i = \bigvee_j \ell_{i,j}$, let $L_i$ be the set of vertices $v_{\ell,\ell'}$ where neither $\ell$ nor $\ell'$ appears in $C_i$. 

We create two new sets $W_i,W'_i$ of size respectively $M$ and $M-|L_i|+k-1$ and all the edges between $W_i$ and $W_i'$. This bipartite graph is called the \emph{clause checker} of $C_i$. 

We finally connect all the vertices of $W_i$ to all vertices of $L_i$, and put a guard on each vertex of $W'_i$.
\item We create a last guarded vertex $s$.
\item We add an edge between every pair of guarded vertices.
\end{itemize}

\begin{figure}[!ht]
\begin{center}
\scalebox{0.8}{
	\begin{tikzpicture}[
every edge/.style = {draw=black,very thick},
 vertex/.style args = {#1/#2}{circle, draw, thick, fill=black,
      label=#1:#2}
                    ]

  \draw [dashed,color=red] (-1,1) -- (12.5,1);
\draw [dashed,color=red] (7.5,7) -- (12.5,7);
    \draw[ultra thick, color=blue] (-0.7,5.0) rectangle (6.7,8.7);

    \draw[pattern=north west lines, pattern color=blue!30] (-0.5,1.2) rectangle (0.5,4.5);
    \draw[pattern=north west lines, pattern color=blue!30] (5.5,1.2) rectangle (6.5,4.5);
    \draw[pattern=north west lines, pattern color=blue!30] (-0.5,7.2) rectangle (0.5,8.5);
    \draw[pattern=north west lines, pattern color=blue!30] (5.5,7.2) rectangle (6.5,8.5);

    \node[vertex=below/$s$] (s) at (10, 4) {};

    \node[vertex=below/$u_{x_1}$,fill=white] (x1) at (0,0) {};
    \node[vertex=below/$u_{\overline{x_1}}$,fill=white] (x1b) at (2,0) {};

    \node[vertex=below/$u_{x_2}$,fill=white] (x2) at (4,0) {};
    \node[vertex=below/$u_{\overline{x_2}}$,fill=white] (x2b) at (6,0) {};

    \node[vertex=below/$v_{x_1,y_1}$] (y11) at (0,2) {};
    \node[vertex=below/$v_{x_1,\overline{y_1}}$] (y11b) at (0,4) {};

    \node[vertex=below/$v_{x_1,y_2}$] (y12) at (0,6) {};
    \node[vertex=below/$v_{x_1,\overline{y_2}}$] (y12b) at (0,8) {};

    \node[vertex=below/$v_{\overline{x_1},y_1}$] (y1b1) at (2,2) {};
    \node[vertex=below/$v_{\overline{x_1},\overline{y_1}}$] (y1b1b) at (2,4) {};

    \node[vertex=below/$v_{\overline{x_1},y_2}$] (y1b2) at (2,6) {};
    \node[vertex=below/$v_{\overline{x_1},\overline{y_2}}$] (y1b2b) at (2,8) {};

    \node[vertex=below/$v_{x_2,y_1}$] (y21) at (4,2) {};
    \node[vertex=below/$v_{x_2,\overline{y_1}}$] (y21b) at (4,4) {};

    \node[vertex=below/$v_{x_2,y_2}$] (y22) at (4,6) {};
    \node[vertex=below/$v_{x_2,\overline{y_2}}$] (y22b) at (4,8) {};

    \node[vertex=below/$v_{\overline{x_2},y_1}$] (y2b1) at (6,2) {};
    \node[vertex=below/$v_{\overline{x_2},\overline{y_1}}$] (y2b1b) at (6,4) {};

    \node[vertex=below/$v_{\overline{x_2},y_2}$] (y2b2) at (6,6) {};
    \node[vertex=below/$v_{\overline{x_2},\overline{y_2}}$] (y2b2b) at (6,8) {};

    \node[circle, draw, thick] (c1) at (8, 0) {};
    \node[circle, draw, thick] (c2) at (9, 0) {};
    \node at (10, 0) {...};
    \node[circle, draw, thick] (c3) at (11, 0) {};
    \node[circle, draw, thick] (c4) at (12, 0) {};

    \node[circle, draw, thick, fill=black] (c1p) at (8, 2) {};
    \node[circle, draw, thick, fill=black] (c2p) at (9, 2) {};
    \node at (10, 2) {...};
    \node[circle, draw, thick, fill=black] (c3p) at (11, 2) {};
    \node[circle, draw, thick, fill=black] (c4p) at (12, 2) {};

    \draw [ultra thick,decorate,decoration={calligraphic brace,mirror}] (7.6,-0.5) --  (12.4,-0.5);

    \node at (10, -1){$V_{y_2}$: $M$ vertices};
    \node at (10, 3){$V'_{y_2}$: $M-4k+1$ vertices};

    \draw [ultra thick,decorate,decoration={calligraphic brace}] (7.6, 2.5) --  (12.4,2.5);
    
    \node[circle, draw, thick] (d1) at (8, 6) {};
    \node[circle, draw, thick] (d2) at (9, 6) {};
    \node at (10, 6) {...};
    \node[circle, draw, thick] (d3) at (11, 6) {};
    \node[circle, draw, thick] (d4) at (12, 6) {};

    \node[circle, draw, thick, fill=black] (d1p) at (8, 8) {};
    \node[circle, draw, thick, fill=black] (d2p) at (9, 8) {};
    \node at (10, 8) {...};
    \node[circle, draw, thick, fill=black] (d3p) at (11, 8) {};
    \node[circle, draw, thick, fill=black] (d4p) at (12, 8) {};

    \draw [ultra thick,decorate,decoration={calligraphic brace,mirror}] (7.6,5.5) --  (12.4,5.5);

    \node at (10, 5){$F_2$: $M$ vertices};
    \node at (10, 9){$F'_2$: $M-|L_2|+k-1$ vertices};

    \draw [ultra thick,decorate,decoration={calligraphic brace}] (7.6, 8.5) --  (12.4,8.5);
    
    \path
    (x1) edge (x1b)
    (x2) edge (x2b)

    (x1) edge[bend left] (y11)
    (x1) edge[bend left] (y11b)
    (x1) edge[bend left] (y12)
    (x1) edge[bend left] (y12b)

    (x1b) edge[bend left] (y1b1)
    (x1b) edge[bend left] (y1b1b)
    (x1b) edge[bend left] (y1b2)
    (x1b) edge[bend left] (y1b2b)

    (x2) edge[bend left] (y21)
    (x2) edge[bend left] (y21b)
    (x2) edge[bend left] (y22)
    (x2) edge[bend left] (y22b)

    (x2b) edge[bend left] (y2b1)
    (x2b) edge[bend left] (y2b1b)
    (x2b) edge[bend left] (y2b2)
    (x2b) edge[bend left] (y2b2b)

    (c1) edge (c1p)
    (c1) edge (c2p)
    (c1) edge (c3p)
    (c1) edge (c4p)
    (c2) edge (c1p)
    (c2) edge (c2p)
    (c2) edge (c3p)
    (c2) edge (c4p)
    (c3) edge (c1p)
    (c3) edge (c2p)
    (c3) edge (c3p)
    (c3) edge (c4p)
    (c4) edge (c1p)
    (c4) edge (c2p)
    (c4) edge (c3p)
    (c4) edge (c4p)

    (d1) edge (d1p)
    (d1) edge (d2p)
    (d1) edge (d3p)
    (d1) edge (d4p)
    (d2) edge (d1p)
    (d2) edge (d2p)
    (d2) edge (d3p)
    (d2) edge (d4p)
    (d3) edge (d1p)
    (d3) edge (d2p)
    (d3) edge (d3p)
    (d3) edge (d4p)
    (d4) edge (d1p)
    (d4) edge (d2p)
    (d4) edge (d3p)
    (d4) edge (d4p)
    ;

 	\end{tikzpicture}}
\caption{A partial representation of the reduction from the formula $(x_1 \vee \overline{y_1}) \wedge (\overline{x_1} \vee x_2 \vee y_2)$.
Only two checkers are represented, one for the variable $y_2$ and the other for the clause $C_2 = \overline{x_1} \vee x_2 \vee y_2$.
The vertices above the dashed line are guarded and forms a clique.
All vertices of $V_{y_2}$ are connected to the 8 vertices in the thick rectangle.
All vertices of $F_2$ are connected to the vertices in the hatched zones.
The length of the game is $t = k+M = 34$.
}
\label{fig:pspace}
\end{center}
\end{figure}
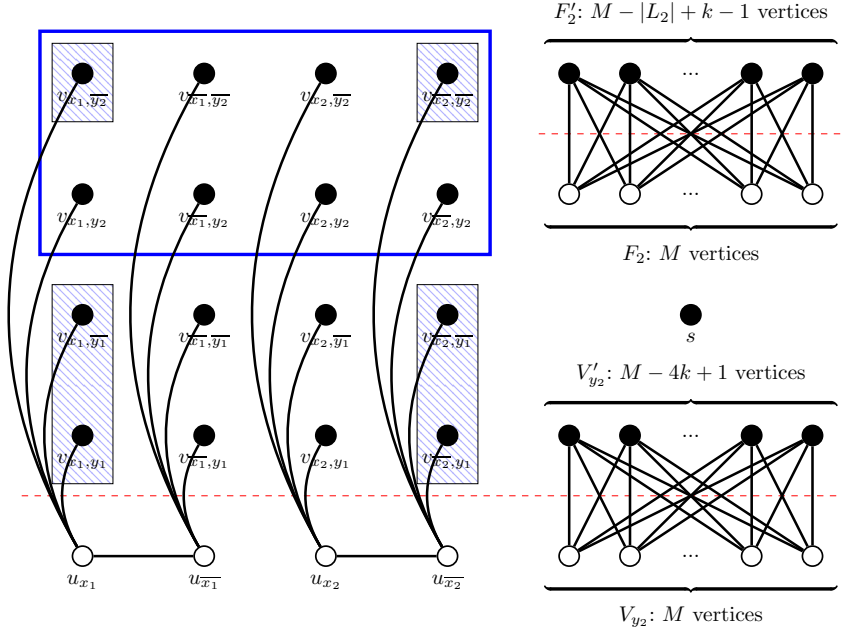

The graph $G_\varphi$ is $2$-unipolar since removing the clique induced by the guarded vertices leaves isolated vertices or isolated edges. Moreover, the construction can be carried out in polynomial time. It remains to show that the reduction is correct. We split this result in two lemmas.

\begin{lemma}
If Falsifier wins \textsc{Unordered-CNF} on $\varphi$, then Attacker wins the mobile domination game on $(G_\varphi, \guards_\varphi)$ in at most $k+M$ turns.
\end{lemma}

\begin{proof}
Assume that Falsifier has a winning strategy on $(X,Y,\varphi)$. We construct a winning strategy for Attacker in $k+M$ turns on $(G_\varphi,\guards_\varphi)$.

The strategy of Attacker will consist in only attacking vertices of $U$ until he has a simple winning strategy. One of them is the following:

\begin{itemize}
    \item If at most $k$ attacks were already performed and if there exists $y$ such that $V_y^*$ contains at most $4k-2$ guards then Attacker successively attacks the vertices of $V_y$. 
\end{itemize}
Since $V_y'$ has size $M-4k+1$  and $V_y$ has size $M$, $V_y$ is defended by at most $M-4k+1+(4k-2) < M$ guards. And then Attacker has a winning strategy in at most $M$ rounds. So we can assume that, during the $k$ first steps, if Attacker chooses a vertex of $U$, Defender should move at most one vertex in each set $V_y^*$. 

During the first $k$ turns, Attacker will either play the strategy described above or attack a vertex $u$ in $\{ u_x,u_{\overline x} \}$ for some $x \in X$. By construction of the instance, Defender can only defend such an attack by moving a vertex from some vertex $v_{\ell,\ell'}$ to $u$ with $\ell\in\{x,\bar{x}\}$. 

Let us denote respectively by $u_{a_i}$ and $v_{a'_i,b_i}$ the attacked vertex at step $i$ and the vertex where the guard defending the attack was posted.
Note that because of the strategy explained above, we can assume that, for every $i,j \le k$, $b_i \ne b_j$ and $b_i \ne \overline{b_j}$ (otherwise, two vertices of some $V_y^*$ were slid to $U$ and then $V_y^*$ contains at most $4k-2$ guards).
So the sequence of the $j$ first rounds provides a partial truth assignment where $j$ variables of $X$ are given a truth value as well as $j$ variables of $Y$. 

So the strategy of Attacker during the first $k$ rounds consists either in playing the simple strategy if he can or following the strategy of Falsifier in $(X,Y,\phi)$ by playing on $u_x$ whenever Falsifier sets the variable $x$ to true in the partial truth assignment described above.

Since Falsifier wins on $(X,Y,\varphi)$, the resulting truth assignment does not satisfy some clause $C_i$ of $\varphi$. In particular, during each of the $k$ first turns, Defender moved a guard from $L_i$ to some $u_\ell$, so only $|L_i|-k$ guards remain on $L_i$. Now the $M$ vertices of $W_i$ are defended by at most $|L_i|-k+|W'_i|=M-1$ guards, hence Attacker can win in $M$ turns by successively attacking each vertex of $W_i$. 
\end{proof}

It now remains to prove the converse direction restated in the following lemma.

\begin{lemma}\label{lem:hard_converse}
If Satisfier wins \textsc{Unordered-CNF} on $\varphi$, then Defender can resist $k+M$ attacks in the mobile domination game on $(G_\varphi, \guards_\varphi)$.
\end{lemma}

\begin{proof}
Assume that Satisfier has a winning strategy on $\varphi$. We call the guard posted initially on $s$ the \emph{special guard}. The strategy of Defender on $(G_\varphi,\guards_\varphi)$ consists in applying the following rules:
\begin{enumerate}
    \item If Attacker attacks a vertex $v$ that was guarded, then Defender moves the special guard to $v$. Defender never moves the special guard outside $N[s]$ and then the move is possible since the set of guarded vertices form a clique\footnote{One can easily remark that this strategy is compatible with the remaining rules.}.
   \item If Attacker attacks an (initially unguarded) vertex of a  (variable or clause) checker, Defender defends, if she can, by moving a vertex initially guarded in the same checker (hence not the special guard). If she cannot, we say that Attacker \emph{engaged} this checker, and Defender defends by moving a guard from $V^*_y$ (for the variable checker $(V_y,V'_y)$) or $L_i$ (for the clause checker $(W_i,W'_i)$). \\
   Note that, by construction and since guards in checkers will only be used to defend vertices in the checker, Attacker has to attack at least $M-4k+1$ times a variable checker and at least $M-|L_i|+k-1\geqslant M-4k^2$ times a clause checker to engage the checker. Therefore, Attacker can engage at most one checker in $M+k$ turns, by our choice of $M$. Note that a checker can be engaged several times.

  \item If Attacker attacks the vertex $u_\ell$ for some literal $\ell$ and a token is already on $u_{\bar{\ell}}$, then Defender moves the guard from $u_{\bar{\ell}}$ to $u_\ell$. \\
  In the rest of the strategy, if we move a guard on a vertex of the pair $\{ u_\ell,u_{\bar{\ell}} \}$, the guard will stay on that pair of vertices (note that by 1., it cannot be the special guard).
  
\item\label{pspace-r4}  Assume now Attacker attacks the vertex $u_\ell$ for some literal $\ell$ and no token is already on $u_{\bar{\ell}}$ and no variable checker has been engaged. In this case, let $\ell'$ be the literal set to true by Satisfier when Falsifier sets $\ell$ to true (this is a valid move in the \textsc{Unordered-CNF} game by the previous rule). Observe that the guard initially on $v_{\ell,\ell'}$ did not move before, hence Defender defends by moving it.

\item\label{pspace-r5} Finally, assume that Attacker attacks the vertex $u_\ell$ for some literal $\ell$ and no token is already on $u_{\bar{\ell}}$ and a checker has been engaged.
Then Defender answers by moving a non-special guard from some vertex $v_{\ell,\ell'}$ not adjacent to the engaged checker. In particular, this is always possible if the engaged checker is a variable checker since $k\geqslant 2$, or if it is a clause checker since each clause contains a literal from $Y$. \end{enumerate}

Observe that the moves played by the second-to-last rule yield a partial truth assignment for the variables of $X\cup Y$. Moreover, since Satisfier wins, one can extend this assignment so that $\varphi$ becomes true.

Since, by construction, the vertices initially guarded and the vertices $u_\ell$ for every $\ell$ can always be defended during the $M+k$ steps, it simply remains to show that Defender can always answer attacks in an engaged checker. 

Assume that Attacker engaged the variable checker $(V_y,V'_y)$ for the $(r+1)$-th time. Note that after engaging the checker for the first time, no guard from $V_y^*$ will leave $V_y^*\cup V_y\cup V'_y$ by Rule~\ref{pspace-r5}. Moreover, before the checker is engaged, Rule~\ref{pspace-r4} ensures that at most one non-special guard from $V^*_y$ has been moved (to some $u_\ell$). Each time the checker is engaged, one non-special guard from $V^*_y$ is moved to $V_y$. Therefore, in the current situation, there are at least $4k-r-1$ non-special guards in $V^*_y$. Note that the number of unguarded vertices in $V_y$ is $|V_y|-|V_y'|-r=4k-1-r$. Since there is an unguarded vertex in $V_y$ (because Attacker played on it), we have $4k-r-1\geqslant 1$, hence there is a non-special guard on $V^*_y$ that Defender can move.

A similar argument holds when the engaged checker is a clause checker $(W_i,W'_i)$. Indeed, assume that at some point $k$ non-special guards of $L_i$ have been moved to some $u_\ell$. This means that Rule~\ref{pspace-r5} never applied and Rule~\ref{pspace-r4} was applied $k$ times. Moreover, all the literals set to true by Rule~\ref{pspace-r4} do not appear in $C_i$, hence $C_i$ (and furthermore $\varphi$) is not satisfied, a contradiction since Defender played following Satisfier's winning strategy on $(X,Y,\varphi)$. Therefore, at most $k-1$ non-special guards of $L_i$ have been moved to some $u_\ell$, so at least $|L_i|-k+1$ remain at any time. Therefore, Defender can answer attacks on $W_i$ by playing a vertex in $L_i$ every time there is no guard in $W'_i$.

\end{proof}

\section{Trees}
\label{sec:tree}

The goal of this section consists in proving the following theorem:

\begin{theorem}
    \label{thm:treepoly}
    $\kED$ can be solved in polynomial time on trees.
\end{theorem}

In a first part, given $G$ and $\guards$, we will introduce the concept of arenas, and we will show that, when $G$ is a tree, $t_G(D)$ is related to the maximum possible treedepth of an arena of $G$. We will then prove that this parameter is exactly the parameter of interest to decide $\kED$ on trees. We then recall an algorithm from Iyer et al.~\cite{iyer88} to compute the treedepth of a tree, and adapt it in a last part to compute $t_G(D)$ in polynomial time.

\subsection{The right parameter}

The main concepts introduced in this section are arenas. Roughly speaking, an arena $X$ of a tree $T$ and a set of guards $\guards$ is a minimal subtree of $T$ where Attacker wins and where guards outside $X$ cannot help the defense of $X$. 

\begin{definition}\label{def-arena}
Let $T$ be a tree and $\guards \subseteq V(T)$ be a set of guards.
An \emph{arena} of $(T, \guards)$ is a subtree $X$ of $T$ such that:
\begin{enumerate}
\item\label{a_p1} $V(X)\cap \guards$ and $V(X)\setminus \guards$ are both independent;
\item\label{a_p2} the guarded vertices in $X$ have degree exactly 2 in $X$ (in particular no leaf is guarded);
\item\label{a_p3} there is no vertex in $V(X) \setminus \guards$ with a neighbor in $\guards \setminus V(X)$.
\end{enumerate}
\end{definition}

For the ease of notation, we identify a subtree $X$ with its set of vertices, so that $X\cap \guards$ denotes the set of guarded vertices in $X$. Equivalently (even if we will not prove it since we will not use it), an arena is a set satisfying the following:
\begin{itemize}
    \item $|\guards \cap X| < \alpha(T[X])$;
    \item there is no unguarded vertex in $X$ with a guarded neighbor outside $X$ (that is, there is no edge between $X \setminus \guards$ and $\guards \setminus X$);
    \item $X$ is minimal for these properties.
\end{itemize}

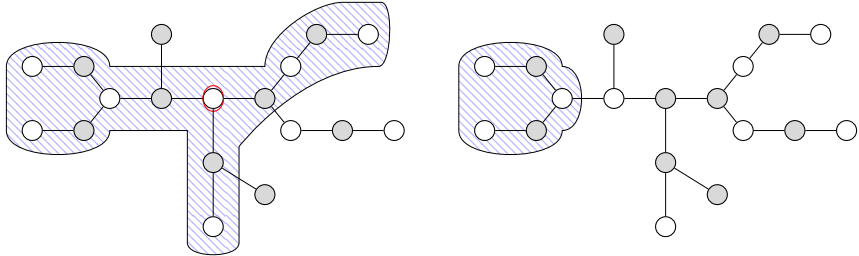
\begin{figure}[!ht]
\begin{center}
\scalebox{0.5}{
	\begin{tikzpicture}[scale=0.85,vertex/.style={circle,draw,thick,fill=white,minimum size=15pt}]

\begin{scope}[shift={(0,0)},xscale=.8]
    \draw [pattern=north west lines, pattern color=blue!30]
(-1,2) .. controls +(0,1) and +(0,1) .. (3,2)
.. controls +(0,0) and +(0,0) .. (9,2)
.. controls +(0,1) and +(-1,0) .. (12,4)
.. controls +(0,0) and +(0,0) .. (13.5,4)
.. controls +(0.5,0) and +(0.5,0) .. (13.4,2)
.. controls +(-3,0) and +(0,0) .. (8,-0.5)
.. controls +(0,0) and +(0,0) .. (8,-3.5)
.. controls +(0,-0.5) and +(0,-0.5) .. (6,-3.5)
.. controls +(0,0) and +(0,0) .. (6,0)
.. controls +(0,0) and +(0,0) .. (3,0)
.. controls +(0,-1) and +(0,-1) .. (-1,0)
.. controls +(0,0) and +(0,0) .. (-1,2);

    \node[vertex] (v0) at (0,0) {};
    \node[vertex] (v1) at (0,2) {};
    \node[vertex,fill=gray!30] (v2) at (2,0) {};
    \node[vertex,fill=gray!30] (v3) at (2,2) {};
    \node[vertex] (v4) at (3,1) {};
    \node[vertex,fill=gray!30] (v5) at (5,1) {};
    \node[vertex,fill=gray!30] (v6) at (5,3) {};
    \node[vertex] (v7) at (7,1) {};
    \node[vertex,fill=gray!30] (v8) at (7,-1) {};
    \node[vertex] (v9) at (7,-3) {};
    \node[vertex,fill=gray!30] (v10) at (9,-2) {};
    \node[vertex,fill=gray!30] (v11) at (9,1) {};
    \node[vertex] (v12) at (10,0) {};
    \node[vertex] (v13) at (10,2) {};
    \node[vertex,fill=gray!30] (v14) at (11,3) {};
    \node[vertex] (v15) at (13,3) {};
    \node[vertex,fill=gray!30] (v16) at (12,0) {};
    \node[vertex] (v17) at (14,0) {};
    \draw[thick,color=red] (7,1) circle(0.4);
    
    \path
    (v0) edge (v2)
    (v1) edge (v3) 
    (v2) edge (v4)
    (v3) edge (v4)
    (v4) edge (v5)
    (v5) edge (v6)
    (v5) edge (v7)
    (v7) edge (v8)
    (v8) edge (v9)
    (v8) edge (v10)
    (v7) edge (v11)
    (v11) edge (v12)
    (v11) edge (v13)
    (v13) edge (v14)
    (v14) edge (v15)
    (v12) edge (v16)
    (v16) edge (v17)
    ;
    
\end{scope}

\begin{scope}[shift={(14,0)},xscale=.8]
    \draw [pattern=north west lines, pattern color=blue!30]
(-1,2) .. controls +(0,1) and +(0,1) .. (3,2)
.. controls +(1,0) and +(1,0) .. (3,0)
.. controls +(0,-1) and +(0,-1) .. (-1,0)
.. controls +(0,0) and +(0,0) .. (-1,2);

    \node[vertex] (v0) at (0,0) {};
    \node[vertex] (v1) at (0,2) {};
    \node[vertex,fill=gray!30] (v2) at (2,0) {};
    \node[vertex,fill=gray!30] (v3) at (2,2) {};
    \node[vertex] (v4) at (3,1) {};
    \node[vertex] (v5) at (5,1) {};
    \node[vertex,fill=gray!30] (v6) at (5,3) {};
    \node[vertex,,fill=gray!30] (v7) at (7,1) {};
    \node[vertex,fill=gray!30] (v8) at (7,-1) {};
    \node[vertex] (v9) at (7,-3) {};
    \node[vertex,fill=gray!30] (v10) at (9,-2) {};
    \node[vertex,fill=gray!30] (v11) at (9,1) {};
    \node[vertex] (v12) at (10,0) {};
    \node[vertex] (v13) at (10,2) {};
    \node[vertex,fill=gray!30] (v14) at (11,3) {};
    \node[vertex] (v15) at (13,3) {};
    \node[vertex,fill=gray!30] (v16) at (12,0) {};
    \node[vertex] (v17) at (14,0) {};
    \path
    (v0) edge (v2)
    (v1) edge (v3) 
    (v2) edge (v4)
    (v3) edge (v4)
    (v4) edge (v5)
    (v5) edge (v6)
    (v5) edge (v7)
    (v7) edge (v8)
    (v8) edge (v9)
    (v8) edge (v10)
    (v7) edge (v11)
    (v11) edge (v12)
    (v11) edge (v13)
    (v13) edge (v14)
    (v14) edge (v15)
    (v12) edge (v16)
    (v16) edge (v17)
    ;
\end{scope}

	\end{tikzpicture}}
\caption{On the left, a tree with an arena of contracted treedepth 3.
Each grayed vertex is guarded, and the attack is circled.
On the right, the same tree after the Defender has moved the guard from left to the attack. A smaller arena of contracted treedepth 2 is created (its contracted tree is a path on three vertices).}
\label{fig:arena}
\end{center}
\end{figure}

\begin{definition}
    Let $T$ be a tree, $D\subset V(T)$ and $X$ an arena of $(T,D)$. The \emph{contracted tree} of $X$ is obtained from $X$ by replacing each vertex from $V(X)\cap D$ by an edge between its neighbors. The \emph{contracted treedepth} of $X$, denoted by $\ctd(X)$, is the treedepth of this contracted tree.
\end{definition}

In the rest of this part, we will prove the following:
\begin{theorem}
\label{thm:win=min}
Given a tree $T$ and a set of guards $\guards$, $t_T(\guards)$ is the minimum value of $\ctd(X)$ over all arenas $X$ of $(T,\guards)$.
\end{theorem}

In particular, Theorem~\ref{thm:win=min} ensures that $t_T(\guards)=+\infty$ if and only if $(T,\guards)$ has no arena. Let $T$ be a tree and  $\guards$ be a set of guards. We split the proof of Theorem~\ref{thm:win=min} into two independent lemmas.

\begin{lemma}\label{tree-ctd-turns}
If $(T, \guards)$ has an arena $X$ of contracted treedepth $t$, then 
$t_T(\guards) \leq t$.
\end{lemma}

\begin{proof}
We proceed by induction on $t$. If $t=1$, then, since in the construction of the contracted tree, we only remove vertices of degree $2$, the contracted tree should initially contain only one vertex $v$. The vertex $v$ should be unguarded by (\ref{a_p2}) and has no guarded neighbor in $T$ by (\ref{a_p3}). Hence, Attacker wins in one turn by attacking $v$.

Assume that $t>1$, and take a td-decomposition $T^*$ of $T[X]$ of depth $t$. The strategy of the Attacker is as follows: he attacks the root $v$ of $T^*$. Denote by $u$ the vertex of $\guards$ moved by Defender (if Defender cannot defend the conclusion follows). Note that $u \in \guards$ and then $u$ has degree $2$ in $T[X]$. Let $w$ be the other neighbor of $u$ in $X$.

Consider $X'$ the component of $T[X] - uw$ containing $w$. Observe that $X'$ is an arena of $(T, \guards \cup \{ v \} \setminus \{ u \})$. Indeed, $X'$ is a subset of $X$ which contains neither $u$ nor $v$ so (\ref{a_p1}) holds. Since $w$ is at even distance from the root, (\ref{a_p2}) holds. Finally, by (\ref{a_p3}) and the fact that $u$ has been moved on $v$ ensures that (\ref{a_p3}) holds.
Moreover, the contracted tree of $X'$ is a subtree of $T^*\setminus\{v\}$, therefore it has depth at most $t-1$ and $\ctd(X')\leqslant t-1$. By induction, Attacker wins on $(T, \guards \cup\{v\}\setminus\{u\})$ in at most $t-1$ turns and thus, wins on $(T, \guards)$ in at most $t$ turns.
\end{proof}

In particular, Attacker wins the eternal domination game if and only if there is an arena. 

\begin{lemma}\label{tree-turns-ctd}
If $t_T(\guards) = t < +\infty$, then there exists an arena of contracted  treedepth at most $t$.
\end{lemma}

\begin{proof}[Proof of Lemma~\ref{tree-turns-ctd}]
We proceed again by induction. If $t=1$, then there exists an unguarded vertex $v$ without neighbor in $\guards$. In particular, it is an arena of contracted treedepth $1$.

Assume that $t>1$, and let $v$ be the first attack of Attacker in a shortest winning strategy. Let $u$ be a guarded neighbor of $v$. Assume that Defender moves the guard of $u$ to answer the attack on $v$. By induction, after the move, Attacker wins in $t-1$ turns at most, hence there exists an arena $X_u$ of contracted treedepth at most $t-1$ in $(T,\guards')$ (where $\guards'=\guards\cup\{v\}\setminus\{u\}$). 

Assume that $v\in X_u$. By definition $v$ has degree 2 in $X_u$, hence it has a neighbor $w\neq u$. The subtree containing $w$ in $X_u\setminus\{vw\}$ is an arena in $(T,\guards)$ of contracted treedepth at most $t-1$, contradicting $t_T(\guards)=t$ by Lemma~\ref{tree-ctd-turns}. Therefore, $v\notin X_u$ and by definition of arenas, $u\notin X_u$ either. 

Note that $X_u$ must contain an unguarded neighbor of $u$, otherwise it is an arena in $(T,\guards)$, again a contradiction with $t_T(\guards)=t$.

Now consider the set $X=\{v\}\cup \bigcup_{u\in N(v)\cap \guards} (X_u\cup\{u\})$. Observe that it is an arena of contracted treedepth at most $t$ (since removing $v$ and its guarded neighbors from $X$ yields components of contracted treedepth at most $t-1$).
\end{proof}

This concludes the proof of Theorem~\ref{thm:win=min}. Note in particular that, since the treedepth of a tree is logarithmic in its number of vertices~\cite{nevsetvril2012sparsity}, we get the following corollary.
\begin{corollary}
If $T$ is an $n$-vertex tree and $D$ a set of guards, then $t_T(D)$ is either $\infty$ or $O(\log n)$.
\end{corollary}

It remains to show that we can compute the minimum contracted treedepth among all arenas in polynomial time. This is the goal of the rest of this section.

\subsection{Ranking lists}

Our algorithm is built upon the algorithm for computing the treedepth of a tree by Iyer et al.~\cite{iyer88}. We start by presenting (a rephrased version of) their algorithm, and stating some of their results about its correctness.

\begin{definition}
Given two non-increasing lists of integers $L_1$ and $L_2$, we define the merged list $L_1 \oplus L_2$ as the concatenation of $L_1$ and $L_2$ sorted by decreasing order.
\end{definition}

\begin{definition}\cite{iyer88}
Let $T$ be a rooted tree. The \emph{ranking list} of $T$  is
\[\rl(T)=\begin{cases} \varnothing &\text{ if $T$ is empty}\\
(\td(T))\oplus \rl(T-T') &\text{ otherwise}
\end{cases}\]
where $T'$ is the unique inclusion-wise minimal rooted subtree of $T$ with $\td(T')=\td(T)$.
\end{definition}

The uniqueness of $T'$ is guaranteed by the following lemma, since two rooted subtrees are either disjoint or contained one in the other.
\begin{lemma}\label{treedepth-disjoint-sets}\cite{iyer88}
Let $T$ be a tree and $V_1, V_2$ be two disjoint subsets of $V(T)$.
If $\td(T[V_1])=\td(T[V_2])=d$, then $\td(T) > d$.
\end{lemma}

The algorithm from Iyer et al. actually computes the ranking list of a given rooted tree $T$ in polynomial time, and then outputs its first element (namely, $\td(T)$). More precisely, they give a procedure RankRoot that computes the ranking list of a tree $T$ from the ranking lists of the subtrees rooted at children of its root. We rephrase their algorithm in a way more suited to our purposes, using two operations: the merge $\oplus$ of two sorted lists, and the closure of a list.

\begin{definition}
    Given a non-increasing list $L$, we define its closure $\cl{L}$ as follows. Let $L_s$ the longest suffix of $L$ which is either empty or satisfies $L_s \geq_{lex} (k, k-1, \ldots, 1)$ where $k$ is the first element of $L_s$. Then, $\cl{L}$ is build from $L$ by replacing $L_s$ by the 1-element list $(k+1)$ (or by $(1)$ if $L_s$ is empty).
\end{definition}

\begin{example}
$\cl(7, 6, 4, 3, 2, 2)=(7,6,5)$ since the largest suffix we can find is $(4, 3, 2, 2)$. Moreover, $\cl{(4, 3, 2)} = (4, 3, 2, 1)$.
\end{example}

The following lemma states the recurrence relation used in~\cite{iyer88} to compute the ranking list of a given tree. 

\begin{lemma}\label{ranking-list-rule}\cite{iyer88}
Let $T$ be a tree rooted in $r$ and $T_1, \ldots, T_k$ the subtrees of $T$ rooted at the children of $r$. Then $\rl(T)=\cl(\rl(T_1)\oplus \cdots \oplus \rl(T_m))$.
\end{lemma}

\subsection{Computing $t_T(\guards)$}

The main part of our algorithm is a recursive procedure that takes as input a tree $T$ rooted at an unguarded vertex and determines the smallest contracted treedepth of an arena containing its root (if such an arena exists). Following the steps of Iyer et al., computing contracted treedepth is not sufficient for the recursion to carry out. We actually compute the \emph{contracted ranking list} of an arena, that is the ranking list of its contracted tree.

To compute $t_T(\guards)$, we just have to apply this procedure to all possible ways of rooting $T$ at an unguarded vertex, and output the minimum contracted treedepth we encountered which can indeed be done in polynomial time. 

First observe that we can simplify the tree $T$ by trimming some parts that will never be in any arena. Namely, we can remove all the edges between two guarded (resp. unguarded) vertices. Moreover we can also remove every guarded leaf of $T$ together with its neighbor. Note that this may disconnect $T$, but since arenas are connected, we can just handle each connected component separately. Without loss of generality, we may thus assume that no such operation can be applied, that is $T$ is rooted at an unguarded vertex $r$, the sets $D$ and $V(T)\setminus D$ are both independent, and all leaves of $T$ are unguarded. In that case, we say that $T$ is \emph{nice}.

Denote by $u_1,\ldots,u_k$ the children of $r$, which must thus be guarded. Moreover, for $1\leqslant i \leqslant k$, let $v_{i,1},\ldots,v_{i,\ell_i}$ be the children of $u_i$ (which are unguarded). Note that any arena of $T$ containing $r$ must contain $u_1,\ldots,u_k$, and for every $1\leqslant i\leqslant k$, exactly one $v_{i,j}$. Moreover, the restriction of $X$ to the subtree $T_{i,j}$ of $T$ rooted at $v_{i,j}$ is an arena of this subtree containing its root. 

Finding an arena of smallest contracted treedepth thus amounts to choosing the right children of each $u_i$. We proceed in a greedy way, each time choosing $v_{i,j}$ with the (lexicographically) smallest contracted ranking list.
Any of these move will force Defender to move a guard from the arena to one of its two neighbors. The previously guarded vertex will then be unguarded and adjacent to another unguarded vertex, therefore, this will create an arena of smaller contracted treedepth in the first one.
With these notations, we obtain the following algorithm.

\begin{algorithm}
\caption{\textsc{Comp-Arena}}\label{algo-comp-arena}
\KwData{A nice tree $T$ rooted at $r$, a set of vertices $\guards$ with $r\notin D$.}
\KwResult{An arena of minimal contracted ranking list (w.r.t. $\leq_{lex}$) among arenas of $(T, \guards)$ containing $r$, together with its contracted ranking list.}
  \eIf{$|V(T)| = 1$}{
    return $\{r\},(r)$ \;
  }{
     \For{$i = 1$ \KwTo $k$}{
        \For{$j = 1$ \KwTo $\ell_i$}{
            $(A_{i,j},L_{i,j}) \leftarrow \textsc{Comp-Arena}(T_{i, j}, \guards \cap V(T_{i,j}))$ \;
        }
        $(A_i,L_i) \leftarrow $ the pair $(A_{i,j},L_{i,j})$ \mbox{with minimum $L_{i,j}$ (w.r.t. $\leq_{lex})$} \;
     }
      return $\{r\} \cup \bigcup_i (A_i + u_i), \cl{(L_1\oplus\cdots \oplus L_k)}$ \;
    }
\end{algorithm}

This algorithm clearly runs in polynomial time. It remains to show that it is correct. Note that the list we output is indeed the contracted ranking list of the arena we output, by Lemma~\ref{ranking-list-rule}. The hard part is to show that the result has actually minimal contracted treedepth. This is a consequence of the following lemma. 

\begin{lemma}\label{arena-lex}
Let $T$ be a rooted tree with a rooted subtree $T_s$, $T'_s$ be a rooted tree
and $T'$ obtained from $T$ by replacing $T_s$ with $T'_s$.
If $\rl(T_s) \leq_{lex} \rl(T'_s)$, then $\rl(T) \leq_{lex} \rl(T')$.
\end{lemma}

Note that by transitivity, it is enough to show the Lemma when $T_s$ is rooted at a children of the root of $T$. Observe then that $\rl(T)=\cl{(\rl(T_s)\oplus L)}$ and $\rl(T')=\cl{(\rl(T'_s)\oplus L)}$ for some list $L$. We can easily see that $\oplus$ is compatible with $\leq_{lex}$ as summarized in the following observation. 

\begin{observation}\label{add-preserverse-lex}
Let $L_1, L_2$ and $L$ be three non-increasing lists such that $L_1 \leq_{lex} L_2$.
Then $L_1 \oplus L \leq_{lex} L_2 \oplus L$.
\end{observation}

Therefore, we only have to show that $\cl$ is also increasing w.r.t. $\leq_{lex}$. 

\begin{lemma}
    Let $L_1, L_2$ be two non-increasing lists such that $L_1 \leq_{lex} L_2$. Then $\cl{L_1} \leq_{lex} \cl{L_2}$.
\end{lemma}

\begin{proof}
Let $P$ be the largest common prefix of $L_1$ and $L_2$, so that $L_1= P,x_1,\ldots$ and $L_2=P,x_2,\ldots$ with $x_1 < x_2$. For $i=1,2$, let $S_i$ be the suffix of $L_i$ used when computing $\cl{L_i}$, that is the longest suffix of $L_i$ such that $S_i\geq_{lex} (k_i,k_i-1,\ldots,1)$ where $k_i$ is the first element of $L_i$.

If $S_1$ does not intersect $P$, then $\cl{L_1}\leq_{lex} P,x_1+1\leq_{lex} L_2$ since $x_2\geqslant x_1+1$. This concludes since $L\leq_{lex} \cl{L}$ for every non-increasing list $L$.

Otherwise, we have $k_1\in P$. In particular, $k_1$ appears in $L_2$ and the suffix of $L_2$ starting at the position $k_1$ is larger than $S_1$, hence than $(k_1,\ldots,1)$. Therefore $\cl{L_2}=P',k_2+1$ where $P'$ is a prefix of $P$, and $\cl{L_1}$ starts with $P',k_2$, hence $\cl{L_1}\leq_{lex} \cl{L_2}$ again.
\end{proof}

\section{Cographs}
\label{sec:cographs}

In this section, we will prove the following.

\begin{theorem}\label{cograph-ked}
$\kED$ is computable in polynomial time on cographs.
\end{theorem}

A cograph is either an isolated vertex, or is obtained by taking the disjoint union (denoted by $\cup$) or the complete join (denoted by $\Join$) of two smaller cographs. This inductive definition naturally provides a representation of each cograph as a decomposition tree called \emph{cotree}, whose leaves are the vertices of the cograph, and the internal nodes are either \emph{join nodes} or \emph{union nodes}\footnote{Note that this decomposition is not necessarily unique but, for algorithmic purposes, we only need that such a decomposition tree exists and can be computed in polynomial time which is indeed the case.}. We can recover the adjacency between the vertices of a cograph $G$ from its cotree as follows: two vertices $uv$ are adjacent in $G$ if and only if their closest ancestor $\lca(u,v)$ in the cotree is a join node. 

Our polynomial time algorithm for $\kED$ on cographs will compute $t_G(\guards)$ using a recursive procedure applied to the cotree of $G$. Note that this is not restrictive since one can compute the cotree of a given cograph in linear time~\cite{cograph-recognition}. We actually consider a more general game, the \emph{eternal domination game with reservists}, defined as follows. This notion is similar to firefighter reserve \cite{firefighter}.

This game is played on a board $(G,\guards,r)$ where $G$ is a graph, $\guards$ is a set of guards and $r \geq 1$ is a number of reservists. Attacker plays as usual by attacking a non-guarded vertex $v$. Now Defender defends by either moving a guard from a neighbor of $v$ to $v$, or by adding a new guard on $v$ and decrease $r$ by one. Attacker wins when $r$ reaches $0$, that is when Defender moves her last reservist on $G$. We denote by $t_G(\guards,r)$ the smallest number of turns required for Attacker to win. In particular, observe that $t_G(\guards,1)=t_G(\guards)$. By convention, we assume that $t_G(\guards, 0) = 0$.

Theorem~\ref{cograph-ked} relies on the two following results, that handle respectively the leaves and the nodes of a cotree. For the base case, observe that Attacker can only win on an isolated vertex if it is not guarded and there is only one reservist, and in that case he wins in one turn. Therefore we get the following. 
\begin{observation}\label{obs:reservist-onevertex}
Let $G$ be a graph containing only an isolated vertex.
Then, $t_G(\guards, r) = 0$ if $r=0$, $1$ if $\guards=\emptyset$ and $r=1$ and $+\infty$ otherwise.
\end{observation}

We now state two recursive formulas satisfied by $t_G(D,r)$ that allow us to compute $t_G(D)$. Note that the computation only uses the values of the form $t_H(D\cap V(H),r)$ for some cograph $H$ whose cotree is a subtree of the cotree of $G$ and for some integer $r$ between $0$ and $|V(G)|$. Therefore, the computation can be done in polynomial time, which proves Theorem~\ref{cograph-ked}.

\begin{proposition}
\label{lem:union-and-join}
Let $G_1,G_2$ be two cographs and $\guards_1,\guards_2$ be sets of guards on $G_1$ and $G_2$. We have:
\begin{align*}t_{G_1 \cup G_2}(\guards_1\cup \guards_2,r)&=\min_{0 \leqslant i \leqslant r} (t_{G_1}(\guards_1,i)+ t_{G_2}(\guards_2,r-i))\\
t_{G_1 \Join G_2}(\guards_1\cup\guards_2,r)&=\min \{ t_{G_1}(\guards_1,r+|\guards_2|), t_{G_2}(\guards_2,r+|\guards_1|)\}.
\end{align*}
\end{proposition}

We start with the easiest ones, stated in the two following lemmas.

\begin{lemma}
\label{lem:union-inf}
Let $G_1,G_2$ be two cographs, $r\geqslant 0$ and $\guards_1,\guards_2$ be sets of guards on $G_1$ and $G_2$. We have:
\[t_{G_1\cup G_2}(\guards_1\cup\guards_2,r) \leq \min_{0 \leqslant i \leqslant r} (t_{G_1}(\guards_1,i)+ t_{G_2}(\guards_2,r-i)).\]
\end{lemma}

\begin{proof}
Let $G=G_1\cup G_2$ and $r_1,r_2$ be two integers. If, for $i=1,2$, Attacker wins on $(G_i,\guards_i,r_i)$ in $t_i$ turns, then he can also win on $(G,\guards,r_1+r_2)$ in $t_1+t_2$ turns by first playing $t_1$ turns on $(G_1,\guards_1,r_1)$ (so that Defender is forced to play her $r_1$-th reservist), and then playing $t_2$ turns on $(G_2,\guards_2,r_2)$ so that Defender plays her $(r_1+r_2)$-th reservist. 
\end{proof}

\begin{lemma}
\label{lem:join-inf}
Let $G_1,G_2$ be two cographs, $r\geqslant 0$ and $\guards_1,\guards_2$ be sets of guards on $G_1$ and $G_2$. We have:
\[t_{G_1\Join G_2}(\guards_1\cup \guards_2,r) \leq \min \{ t_{G_1}(\guards_1,r+|\guards_2|), t_{G_2}(\guards_2,r+|\guards_1|)\}.\]
\end{lemma}

\begin{proof}
Let $G=G_1\Join G_2$. Observe that if Attacker plays only on $G_1$, the board is equivalent to $(G_1,\guards_1,r+|\guards_2|)$ since Defender can use the guards of $\guards_2$ as reservists. In particular, attacking only in $G_1$ or only in $G_2$ gives a winning strategy for Attacker in $\min\{ t_{G_1}(\guards_1,r+|\guards_2|), t_{G_2}(\guards_2,r+|\guards_1|)\}$ turns. 
\end{proof}

To prove the remaining inequalities, we consider an arbitrary order $\leq_G$ over the vertices of $G$ and introduce $\mathcal{S}^*$ as the following strategy of Defender: if Attacker attacks a vertex $v$, then Defender answers (if possible) with a guarded vertex $u\in N(v)$ whose common ancestor $\lca(u,v)$ in the cotree of $G$ is the deepest. If several such vertices exist, she picks the smallest with respect to the order $\leq_G$. If no such vertex exists (\emph{i.e.} if $v$ has no guarded neighbor), Defender calls a reservist. We denote by $t_G^*(\guards, r)$ the smallest number of turns required for Attacker to win on $(G, \guards, r)$ against this strategy $\mathcal{S}^*$. Notice that this strategy does not depend on the number of available reservists, nor on the connected components of $G$ that do not contain $v$.

Observe that if Attacker wins in $k$ turns on $(G,\guards,r)$, then he also wins in at most $k$ turns against $\mathcal{S}^*$. Therefore, we get the following.

\begin{observation}\label{obs:opt-inf}
Every cograph $G$ satisfies $t_G^* \leq t_G$.
\end{observation}

In the following, we will show by induction that this inequality is actually an equality, meaning that $\mathcal{S^*}$ is an optimal strategy for Defender. 

For the base case, Observation~\ref{obs:reservist-onevertex} clearly remains true when Attacker plays against the strategy $\mathcal{S}^*$.

\begin{observation}\label{obs:reservist-onevertex-opt}
Let $G$ be a graph containing only an isolated vertex.
Then, $t^*_G(\guards, r) = 0$ if $r=0$, $1$ if $\guards=\emptyset$ and $r=1$ and $+\infty$ otherwise.
\end{observation}

To prove the inductive step, we will show that the converse inequalities of Lemmas~\ref{lem:union-inf} and~\ref{lem:join-inf} hold when replacing $t_G$ by $t_G^*$. More precisely, we show the two following lemmas.

\begin{lemma}
\label{lem:union-opt}
Let $G_1,G_2$ be two cographs, $r\geqslant 0$ and $\guards_1,\guards_2$ be sets of guards on $G_1$ and $G_2$. We have:
\[t_{G_1\cup G_2}^*(\guards_1\cup\guards_2,r) \geq \min_{0 \leqslant i \leqslant r} (t^*_{G_1}(\guards_1,i)+ t^*_{G_2}(\guards_2,r-i)).\]
\end{lemma}

\begin{lemma}
\label{lem:join-opt}
Let $G_1,G_2$ be two cographs, $r\geqslant 0$ and $\guards_1,\guards_2$ be sets of guards on $G_1$ and $G_2$. We have:
\[t^*_{G_1\Join G_2}(\guards_1\cup \guards_2,r) \geq \min \{ t^*_{G_1}(\guards_1,r+|\guards_2|), t^*_{G_2}(\guards_2,r+|\guards_1|)\}.\]
\end{lemma}

Before proving these results, let us show how to use them to conclude the proof of Proposition~\ref{lem:union-and-join}.

\begin{proof}[Proof of Proposition~\ref{lem:union-and-join}]
We prove by induction on cographs that every cograph $G$ satisfies $t_G(\guards, r) = t_G^*(\guards, r)$ for any set of guards $\guards$ and $r\geqslant 0$. The base case is already provided by Observation~\ref{obs:reservist-onevertex} and Observation~\ref{obs:reservist-onevertex-opt}.

Now, let $G$ be a cograph defined as the union of two cographs $G_1$ and $G_2$. Let $r\geqslant 0$ and $\guards$ be a set of guards, and $\guards_i=\guards\cap V(G_i)$ for $i=1,2$. Applying successively Lemma~\ref{lem:union-inf}, the induction hypothesis on $G_1,G_2$ and Lemma~\ref{lem:union-opt}, we get 
\begin{align*}
t_G(\guards,r) &\leq \min_{0 \leqslant i \leqslant r} (t_{G_1}(\guards_1,i)+ t_{G_2}(\guards_2,r-i))\\
&=\min_{0 \leqslant i \leqslant r} (t^*_{G_1}(\guards_1,i)+ t^*_{G_2}(\guards_2,r-i))\leq t^*_G(\guards, r)
\end{align*}
This concludes using Observation~\ref{obs:opt-inf}. The case of joins is similar, we only use Lemmas~\ref{lem:join-inf} and~\ref{lem:join-opt} instead. 
\end{proof}

It thus remains to prove Lemmas~\ref{lem:union-opt} and~\ref{lem:join-opt}. We start with the case of unions.

\begin{proof}[Proof of Lemma~\ref{lem:union-opt}]
Let $G=G_1\cup G_2$, $\guards=\guards_1\cup \guards_2$ and $A$ be a winning strategy in $t^*_G(\guards, r)$ turns on $(G, \guards, r)$ against the strategy $\mathcal{S}^*$. Since the strategy of Defender is prescribed, we may assume that $A$ is just a sequence of moves.

Let $A_1$ (resp. $A_2$) be the subsequence of attacks of $A$ played on $G_1$ (resp. $G_2$) and $r_1$ (resp. $r_2$) be the number of reservists called in $G_1$ (resp. $G_2$). Denote by $a_1$ (resp. $a_2$) the length of $A_1$ (resp. $a_2$), so that $r_1+r_2=r$ and $a_1+a_2=t^*_G(\guards, r)$. Since the strategy  $\mathcal{S}^*$ of Defender only depends on the connected component of the attacked vertex, for $i=1,2$, $A_i$ is a sequence of attacks on $(G_i, \guards_i)$ that moves $r_i$ reservists when Defender plays according to $\mathcal{S}^*$.  In particular, $a_i\geq t^*_{G_i}(\guards_i,r_i)$. 

Consequently, $t^*_G(D, r) = a_1+a_2 \geq t^*_{G_1}(\guards_1,r_1)+ t^*_{G_2}(\guards_2,r_2)$, which concludes.
\end{proof}

We end this section with the case of joins. Lemma~\ref{lem:join-opt} relies on the following result.

\begin{lemma}\label{lemma-part-of-join}
Let $G=G_1\Join  G_2$ be a cograph, $\guards$ be a set of guards and $r\geq 1$ be an integer.
Then, Attacker can win on $(G, \guards, r)$ in $t^*_G(\guards, r)$ steps against the strategy $\mathcal{S^*}$ by only playing on $G_1$ or by only playing on $G_2$.
\end{lemma}

Lemma~\ref{lemma-part-of-join} states that (up to exchanging $G_1$ and $G_2$) Attacker has a winning strategy in $t^*_{G_1\Join G_2}(\guards,r)$ turns on $(G_1\Join G_2,\guards,r)$ against the strategy $\mathcal{S}^*$ where he plays only on $G_1$. Since vertices of $G_2$ have the highest possible common ancestors with vertices of $G_1$ in the cotree of $G_1\Join G_2$, this is equivalent to playing on $(G_1,\guards_1,r+|\guards_2|)$ against the strategy $\mathcal{S}^*$ where $D_i=D\cap V(G_i)$ for $i=1,2$. In particular, $t^*_G(D,r) \geqslant t_{G_1}^*(\guards_1,r+|\guards_2|)$, and we get Lemma~\ref{lem:join-opt} by symmetry.

The proof of Lemma~\ref{lemma-part-of-join} is based on this trade-off between increasing the number of reservists and playing only on one side of a join. We first summarize this in the following observation.

\begin{observation}\label{join=reservists}
    Let $G_1$ and $G_2$ be two cographs, each one having a set of guards $\guards_1$ and $\guards_2$. Assume that Attacker wins in $k$ turns against $\mathcal{S}^*$ on $(G_1\Join G_2, \guards_1\cup\guards_2,r)$ by playing only on $G_2$. Then $k\geqslant t^*_{G_2}(\guards_2,r+|\guards_1|)$.
\end{observation}

Another observation is given below. It formalizes the intuition that defending using a reservist instead of a guard already in a graph does not slow down Attacker. 

\begin{lemma} \label{reservists>guards}
    Let $G$ be a graph, $\guards$ a set of guards on $G$ and $r \ge 1$. For every $x\in \guards$, we have $t^*_G(\guards,r) \le t^*_G(\guards\setminus\{x\},r+1)$. 
\end{lemma}

\begin{proof}[Proof of Lemma~\ref{reservists>guards}]
We proceed by induction on $t^* = t^*_G(\guards\setminus\{x\},r+1)$. 
If $t^*=1$, then we must have $r=0$ and in that case $t_G^*(D,0)=0$, which concludes.

Assume now that $t^* \ge 2$. Let $(D,r)$ be the current configuration and $x\in D$. We denote by $P_1 = (D,r)$ and $P_2 = (D\setminus \{x\}, r+1)$. Let $u$ be the first vertex played by Attacker in $P_2$, and denote by $P'_2$ the configuration obtained after Defender answered using $\mathcal{S}^*$. Similarly, denote by $P'_1$ the configuration obtained after Defender used $\mathcal{S}^*$ to defend $u$ in $P_1$.  By definition, observe that $t^*=t_G^*(P_2)=1+t_G^*(P'_2)$ and that $t_G^*(P_1)\leqslant 1+t_G^*(P'_1)$. Consider the following cases:

\begin{itemize}

\item Assume that Defender moves the guard from $x$ to defend $u$ in $P_1$, and that $u$ has no guarded neighbor in $P_2$. So Defender uses a reservist in $P_2$ and then the resulting instances $P'_1$ and $P'_2$ are $(D \cup \{u\} \setminus \{x\},r)$ so the conclusion follows immediately. 
    \item Assume that Defender moves the guard from $x$ to defend $u$ in $P_1$ and a guard from a vertex $v$ to defend $u$ in $P_2$. Then $P'_1=(D\setminus\{x\} \cup \{u\}, r)$ and $P'_2=(D\setminus\{x,v\} \cup \{u\},r+1)$. Applying the induction hypothesis with $D$ replaced by $D\setminus\{x\}\cup\{u\}$ yields $t^*_G(P'_1)\leqslant t^*_G(P'_2)$ and we get \[t_G^*(P_1)\leqslant 1+t^*_G(P'_1)\leqslant 1+t^*_G(P'_2) =t^*.\]

    \item Otherwise, Defender defends with strategy $\mathcal{S}^*$ against the attack on $u$ in $P_1$ without moving the guard from vertex $x$. In this case, the strategy $\mathcal{S}^*$ performs the same move in $P_2$, and we may conclude applying induction as in the previous item. \qedhere 
\end{itemize}    
\end{proof}

We are now ready to conclude the proof of Lemma~\ref{lemma-part-of-join}.

\begin{proof}[Proof of Lemma~\ref{lemma-part-of-join}]

We proceed by induction on $t^* = t^*_G(\guards , r)$. If $t^*=1$, the result follows since Attacker wins against $\mathcal{S}^*$ in a single move.

Assume now that $t^* \ge 2$, and let $u$ be the first vertex attacked by Attacker in a winning strategy in $t^*$ turns against $\mathcal{S}^*$. Without loss of generality, assume that $u$ is in $G_1$ and let $v$ be the answer from Defender according to $\mathcal{S}^*$. On the obtained instance $(G, \guards \setminus\{v\} \cup \{u\},r)$ (or $(G, \guards \cup \{u\},r-1)$ if a reservist has been used), Attacker has a winning strategy in $t^*-1$ turns when Defender uses $\mathcal{S}^*$. So, by induction, Attacker has a winning strategy against $\mathcal{S}^*$ playing only in $G_1$ or in $G_2$. We can assume that it is in $G_2$ (since otherwise we are done), and we may thus apply Observation~\ref{join=reservists}.

Setting $\guards_i=\guards\cap V(G_i)$, we consider three cases depending on $v$:
\begin{itemize}
    \item If $v$ is a reservist, then $t^*-1\geqslant t^*_{G_2}((\guards \cup\{u\})\cap V(G_2), r-1+|(\guards\cup\{u\})\cap V(G_1)|)=t^*_{G_2}(\guards_2,r+|\guards_1|)$.
    \item If $v$ is in $G_1$, we similarly get $t^*-1\geqslant t^*_{G_2}(\guards_2,r+|\guards_1|)$. 
\item Otherwise $v$ is in $G_2$ and we get $t^*-1\geqslant t^*_{G_2}(\guards_2\setminus\{v\}, r+|\guards_1|+1)$. By Lemma~\ref{reservists>guards}, this is at least $t^*_{G_2}(\guards_2,r+|\guards_1|)$. 
\end{itemize}

In each case, there is a winning strategy for Attacker on $(G_2,\guards_2,r+|\guards_1|)$ against $\mathcal{S}^*$ in at most $t^*-1$ turns. But this is also a strategy for Attacker on $(G,\guards,r)$ against $\mathcal{S}^*$ where he plays only on $G_2$. This is a contradiction by definition of $t^*$, which concludes the proof. 
\end{proof}

\section{Parameterized complexity}
\label{sec:param}
In this section, we study the parameterized complexity of $\kED$.
We consider two parameters: $t$ the number of turns and $g$ the number of guards. We first consider generic results, then study restricted classes of graphs. 

\subsection{Generic graphs}

On generic graphs, we get the two following results.
\begin{theorem}\label{ked-in-exptime}
$\kED$ parameterized by $g$ is in \XP.
\end{theorem}

\begin{proof}
Let $G = (V, E)$ be a graph. We will compute all the values $t_G(\guards)$ for each subset $\guards\subset V$ of size $g$ in time $O(n^{2g+2})$. 

We first build an auxiliary directed graph $\mathcal{G}$ that models the moves available for each player. For every subset $\guards$ of $g$ vertices of $G$, and for every vertex $v\in V$, $\mathcal{G}$ has a vertex labeled by $\guards$, and one labeled by $(\guards,v)$. Moreover, $\mathcal{G}$ contains the arcs:
\begin{itemize}
    \item $\guards\to(\guards,v)$ for every set $\guards$ and $v\in V$, and
    \item $(\guards,v)\to (\guards\cup\{v\})\setminus \{u\}$ for every set $\guards$, $v\in V\setminus\guards$ and $u\in \guards\cap N(v)$. 
\end{itemize} 

We now label all vertices $(\guards, v)$ such that $D \cap N(v) = \emptyset$ with 0. Then we apply the following rules while it is possible.
\begin{itemize}
\item We label each non-labeled vertex $\guards$ with $\ell+1$ where $\ell$ is the minimum label of its out-neighbors (if at least one such out-neighbor is labeled). 
\item If every out-neighbor of an non-labeled vertex $(\guards, v)$ is labeled, we label it with the maximum label of its out-neighbors.
\end{itemize}
At the end, all non-labeled configurations are labeled $+\infty$.

Note that each rule can be applied in linear time with respect to $\mathcal{G}$. Moreover, each application of these rules labels at least one new vertex (and these labels do not change afterward), hence the total procedure is at most quadratic in $\mathcal{G}$. Since $\mathcal{G}$ has $O(n^{g+1})$ vertices and $O(n^{g+1})$ arcs, this yields an \textsf{XP} algorithm. 

We now claim that afterward, each vertex $\guards$ is labeled with $t_G(\guards)$. Indeed, if $\guards$ is labeled with $+\infty$, then so are all its out-neighbors and in particular, for every $v\notin\guards$, $(\guards,v)$ has an out-neighbor $(\guards\cup\{u\})\setminus v$ labeled with $+\infty$. Therefore, Defender can eternally defend by answering to Attacker in such a way the set of guards stays labeled with $+\infty$. 

Similarly, one can easily see by induction on $\ell$ that if a set $\guards$ of guards is labeled with $\ell$, then Attacker wins in at most $\ell$ turns and Defender can resist to $\ell-1$ attacks on $\guards$, hence $t_G(\guards)=\ell$. The main point of the induction is that at any moment of the game, Attacker has a move available from a position $D$ to a position $(D,v)$ which decreases the label by one, while Defender has a move available from a position $(D,v)$ to a position $D'$ which has the same label.
\end{proof}

\begin{theorem}\label{w1-k-g}
$\kED$ is $\W[1]$-hard when parameterized by $t+g$.
\end{theorem}

\begin{proof}
Let $(G, k)$ be an instance of $\textsc{Independent-Set}$.
We build an instance of \textsc{Fast-Strategy} as follows. Let $G'$ be the graph obtained from $G$ by adding a set $\guards$ of $k-1$ guarded vertices $u_1 \ldots, u_{k-1}$, each of them connected to all vertices of $G$. We claim that $(G, k)$ is a positive instance of $\textsc{Independent-Set}$ if and only if $(G', \guards, k)$ is a positive instance of $\kED$.

If $G$ admits an independent set $I$ of size $k$, then Attacker wins in $k$ turns by attacking successively the vertices in $I$. Indeed, at each turn, Defender has to move a guard from $\guards$ (she cannot move a guard twice since $I$ is independent). Therefore, she loses at the $k$-th turn since no guard is available in $\guards$.

Conversely, assume that Attacker can win in $k$ turns. If Defender manages to move the same guard twice during the game, then either she can answer the last attack with a guard in $G$, or one guard is still on $\guards$ at the beginning of the $k$-th turn, and she can use it to answer the last attack, a contradiction. Therefore, the $k$ vertices played by Attacker must induce an independent set of $G$, which concludes.
\end{proof}
\subsection{Bipartite graphs}

It is easily seen that the reduction in Theorem \ref{conp-hard-bipartite} is an \FPT-reduction when parameterized by $t$ (but not by $g$). Since \textsc{Independent-Set} is $W[1]$-complete when parameterized by the size of the independent set, we obtain the following result.

\begin{theorem}\label{w1-k-bipartite}
$\kED$ parameterized by $t$ is co-$W[1]$-hard on bipartite graphs.
\end{theorem}

The behavior of $\kED$ becomes quite different when parameterized by $g$ instead, as shown by the following.

\begin{theorem}\label{fpt-g-bipartite}
$\kED$ parameterized by $g$ is \FPT{} on bipartite graphs. 
\end{theorem}

\begin{proof}
Let $(G, \guards, t)$ be an instance of $\kED$ where $G$ is a bipartite graph $(A \cup B, E)$. By Lemma \ref{bipartite-oneside} and Lemma \ref{guardsEqOnePart}, we can assume that $B = \guards$ and that Attacker plays only in $A$. In particular, $t_G(\guards)\leqslant |B|+1=g+1$.

We build a kernel $(G', \guards, t)$ of $(G, \guards, t)$ as follows: for every set of at least $g+1$ twins in $A$, delete all but $g+1$ of them. Observe that $G'$ has at most $g+(g+1)2^g$ vertices. We claim that $t_G(\guards) = t_{G'}(\guards)$.

Recall that Attacker always has a shortest winning strategy where he plays only on initially unguarded vertices by Lemma~\ref{bipartite-oneside}. In particular, if he wins in $t$ turns on $(G',\guards)$, he can just use the same strategy to win in $t$ turns on $(G,\guards)$. 

Conversely, if Attacker wins in $t$ turns on $(G,\guards)$, then mimicking his strategy on $(G',\guards)$ is also making him win in at most $t$ turns (note that each time he plays on a vertex $v$ of $G$ that was deleted, there is at least one unguarded twin of $v$ in $G'$). 
\end{proof}

Finally, we consider the problem $\kED$ parameterized by $n-g$, the number of non-guarded vertices, and show that it is \FPT{} (and even has a linear kernel).

\begin{theorem}\label{kernel-n-g}
$\kED$ admits a kernel of size $2(n - g)$ on bipartite graphs.
\end{theorem}

Our proof relies on the following useful lemma, inspired by crown decompositions~\cite{param-book} that may be of independent interest. 

\begin{lemma}\label{bipartite-kernel-lemma}
Let $G = (A \cup B, E)$ be a bipartite graph.
Assume that there is no isolated vertex in $B$ and $|A| < |B|$.
Then one can compute in polynomial time a non-empty set $A' \subseteq A$ and a set $B' \subseteq B$ such that
\begin{itemize}
    \item there is a matching $M$ between $A'$ and $B'$ that covers $A'$;
    \item there is no edge between $B'$ and $A \setminus A'$.
\end{itemize}
\end{lemma}

\begin{proof}
By König's theorem \cite{diestel-book}, there exist a matching $M$ and a vertex cover $C$ such that $|C| = |M| \leq |A|$. Set $A' = C \cap A$ and $B' = B \setminus C$. Observe that $A'\neq\varnothing$, otherwise $C\subset B$, and since $B$ has no isolated vertex, $C=B$, a contradiction since $|A|\geq|C|=|B|>|A|$. 

Now observe that every edge $e$ in $M$ has exactly one endpoint in $C$, hence the edges of $M$ intersecting $A'$ have their other endpoint in $B'$. Thus, these edges form a matching between $A'$ and $B'$ that covers $A'$.

Finally, by definition of a vertex cover, $(A \cup B) \setminus C$ is an independent set of $G$. In particular, there is no edge between $B'=B \setminus C$ and $A \setminus A'=A \setminus C$.
\end{proof}

\begin{proof}[Proof of Theorem~\ref{kernel-n-g}]
Let $G = (A \cup B, E)$ be a bipartite graph and $\guards$ be a set of guards.
By Lemma \ref{bipartite-oneside} and Lemma \ref{guardsEqOnePart}, we can assume $B = \guards$ and Attacker plays only in $A$.
We use the following two reduction rules.
\begin{itemize}
    \item Remove isolated vertices in $B$.
    \item If $|A|<|B|$, let $A',B'$ the sets obtains by Lemma~\ref{bipartite-kernel-lemma}. Then remove $A'$ and $B'$ from $G$ and remove $B'$ from $\guards$.
\end{itemize}

After iterating these two rules while we can, we end up with a kernel $(G',D')$ with less guarded vertices than unguarded, hence it has size at most $2(n-g)$, as requested. It remains to show that these rules are correct. The first rule clearly preserves $t_G(D)$ since guarded isolated vertex cannot be attacked or used to defend another vertex. 

Consider now the second rule, and assume that $|A| < |B|$. Let $A',B'$ the set obtained by applying Lemma~\ref{bipartite-kernel-lemma}, and $(G',D')$ the new instance. We claim that $t_G(\guards) = t_{G'}(\guards')$.

If Attacker has a winning strategy in $t$ turns on $(G', \guards')$, then he can apply the same strategy on $(G, \guards)$ since, by construction, Defender has the same available moves on $(G, \guards)$ and on $(G', \guards')$. Therefore, $t_G(\guards) \leq t_{G'}(\guards')$. 

Conversely, assume that Defender has a strategy to resist to $t$ attacks on $(G', \guards')$. Then she can resist $t$ attacks on $(G,\guards)$ using the following: if Attacker plays on a vertex in $A \setminus A'$, then she uses the same strategy as on $(G', \guards')$. Otherwise, let $M$ be a matching between $A'$ and $B'$ saturating $A'$. When Attacker plays on a vertex $v$ in $A'$, Defender answers with the guard on $B'$ matched with $v$ in $M$.
\end{proof}

\subsection{First-order definability}

In this subsection we show that $\kED$ can be expressed with a first-order formula. 

\begin{lemma}\label{fo-definable}
For every $t>0$, there is a first-order formula $\varphi_t(X)$ such that
$G \models \varphi_t(\guards)$ if and only if $t_G(\guards) \leq t$.
\end{lemma}

\begin{proof}
Let $X$ be a set of guards and $k<t$. Our goal is to define a formula $\psi_{t,k}(X,a_1,d_1,\ldots,a_k,d_k)$ that holds if Attacker can win in at most $t$ turns when the first $k$ turns were already played on $a_1, d_1, \ldots, a_k, d_k$. In particular, we will have $\varphi_t=\psi_{t,0}$. 

First observe that the set $X_k$ of guards obtained from $X$ after playing these $k$ turns can be defined in \FO. Indeed, the formula $\mbox{guards}_k$ defined by $\mathrm{guards}_0(X,x)=x\in X $ and 
\begin{align*}
    \mathrm{guards}_k&(X , a_1, d_1, \ldots, a_k, d_k, x)=\\ &
(x=a_k) \vee (\mbox{guards}_{k-1}(X, a_1, d_1, \ldots, a_{k-1}, d_{k-1}, x) \land x\neq d_k)\end{align*}
holds if and only if $x\in X_k$.

Observe that Attacker wins in at most $t$ turns after $k<t$ turns are played if either $X_k$ is not a dominating set or $k+1<t$ and there is a position $a_{k+1}$ such that for every answer $b_{k+1}$ of Defender, $\psi_{t,k+1}$ holds. Therefore, denoting by 
\[\mathrm{dom}_k(X, a_1, d_1, \ldots, a_k, d_k) = \forall x \exists y (x-y) \wedge \mbox{guards}_k(X, a_1, d_1, \ldots, a_k, d_k, y)\]
the formula stating that $X_k$ dominates $G$, we have $\psi_{t,t-1}=\neg \mathrm{dom}_t$ and 
\begin{align*}
\psi_{t,k}&(X, a_1, d_1, \ldots, a_k, d_k) =\neg \mathrm{dom}_k(X, a_1, d_1, \ldots, a_k, d_k) \\
\vee& \exists a_{k+1} \forall b_{k+1} (a_{k+1}-b_{k+1}\wedge \mathrm{guards}_k(X, a_1,d_1,\ldots,a_k,d_k,b_{k+1}))\\ &\qquad \Rightarrow \psi_{t,k+1}(X, a_1, d_1, \ldots, a_{k+1}, b_{k+1}),    
\end{align*}
which concludes.
\end{proof}

This directly implies complexity upper bounds using some well-known meta-theorems. The first one comes from the complexity of \FO-model-checking in the generic case~\cite{book-flum-grohe}. 

\begin{theorem}\label{param-k-awstar}
$\kED$ parameterized by $t$ is in $\AW[*]$.
\end{theorem}

It is well-known that this complexity drops when considering restricted classes of graphs, for example nowhere-dense graph classes \cite{fo-nowhere-dense} and bounded twin-width graph classes \cite{fo-twinwidth}. Notice that nowhere-dense graph classes include planar graphs, bounded degree graphs, and graphs with an excluded minor~\cite{nowhere-dense}.

\begin{theorem}\label{fpt-nowhere-fpt}
$\kED$ parameterized by $t$ is \FPT{} on nowhere-dense graph classes.
\end{theorem}

One can note that the same holds for bounded twin-width graphs when we are given a contraction sequence.

\section{Conclusion}

We give a summary table of the different complexity results.

\begin{center}
\begin{tabular}{ | c | m{0.27\textwidth} | m{0.27\textwidth} | m{0.27\textwidth} | }
\hline
     & \textsc{Eternal-Domination-Number} & \textsc{Eternal-Dominating-Set} & \textsc{Fast-Strategy} \\
\hline
 Trees & \P \cite{goddard} &  \P (Th \ref{eds-bipartite}) &  \P (Th \ref{thm:treepoly})\\ 
 \hline
 Cographs & \P \cite{goddard}&  \P (Th \ref{cograph-ked}) & \P (Th \ref{cograph-ked})\\
 \hline
 Bipartite & \P \cite{goddard} & \P (Th \ref{eds-bipartite}) & \coNP-hard (Th \ref{conp-hard-bipartite})\\
 \hline
 Split & \P \cite{goddard} & ? & \coNP-hard (Th \ref{conp-hard-split})\\
 \hline
 Perfect & \P \cite{goddard} & ? & \PSPACE-hard (Th \ref{unipolar-pspace})\\
 \hline
 General & \coNP-hard \cite{eds-digraph} & ? & \PSPACE-hard (Th \ref{unipolar-pspace})\\
 \hline
\end{tabular}
\end{center}

In addition, we obtained several results on the parameterized complexity of \textsc{Fast-Strategy}.

\begin{center}
\begin{tabular}{ | c |c | c | c | }
\hline
   class of graphs  & parameter &  complexity \\
\hline
  General & $t$ & $\AW[*]$ (Th \ref{param-k-awstar})\\
\hline
  General & $g$ & \XP (Th \ref{ked-in-exptime})\\
\hline
   General & $g+t$ & \W[1]-hard  (Th \ref{w1-k-g})\\
\hline
   Bipartite & $t$ & co-\W[1]-hard (Th \ref{w1-k-bipartite})\\
\hline
   Bipartite & $g$ & \FPT (Th \ref{fpt-g-bipartite})\\
\hline
   Bipartite &$n - g$ & linear kernel (Th \ref{kernel-n-g})\\
\hline
  Nowhere dense & $t$ & \FPT (Th \ref{fpt-nowhere-fpt})\\
\hline
  Bounded twin-width  & $t + tww$ & \FPT (Th \ref{fpt-nowhere-fpt})\\
\hline

\end{tabular}
\end{center}

Quite often in the literature, game problems are either in \P{} or \PSPACE-hard. Motivated by Theorem~\ref{conp-hard-bipartite}, we propose the following conjecture. 

\begin{conjecture}\label{conj-pspace-hard-bipartite}
$\kED$ is \PSPACE-complete on bipartite graphs.
\end{conjecture}

Besides Conjecture \ref{conj-pspace-hard-bipartite}, we enumerate some open questions for future work.
\begin{enumerate}
    \item Are there other graph classes where \textsc{Fast-Strategy} is polynomial? Good candidates are unit interval graphs. We think that the notion of arena (defined for trees) can be adapted to these graphs.
  \item In this paper, we have only considered the case where there is at most one guard per vertex. One can allow multiple guards on the same vertex \emph{i.e.} the guards now form a multiset. 
  
  There is a generic reduction to the case considered here, roughly consisting in adding twins. More precisely, given a graph $G$ and a multiset of guards $\guards$, consider the graph $G'$ and set of guards $\guards'$ obtained
   by replacing $k$ guards on a vertex $v$ into a set of $k$ true twins, all of them being guarded. One can see that $t_G(\guards) = t_{G'}(\guards')$. Cographs are preserved by this transformation since they are closed under adding true twins. Thus, Theorem \ref{cograph-ked} still holds. However, this is not true for trees and adapting the proof of Theorem~\ref{thm:treepoly} does not seem straightforward.
   
\item For perfect graphs, we know that \textsc{Eternal-Dominating-Number} is in \P{} and \textsc{Mobile-Domination-Game} is \PSPACE-hard. What is the complexity of \textsc{Eternal-Dominating-Set} on such graphs?

    \item From a parameterized complexity perspective, we showed that \textsc{Mobile-Domination-Game} parameterized by the number of guards is \FPT{} on bipartite graphs. It is then natural to look for a polynomial kernel.
   \item Can our results be adapted in the “all guards move” variant of eternal domination (where Defender can move more than one guard at each turn)? 
   An upper bound on the number of moves for Attacker to win is given in~\cite{meds-efficient}. 
   \end{enumerate}

\end{document}